\documentclass[10pt,journal,letterpaper,twocolumn,twoside,nofonttune]{IEEEtran}

\usepackage{graphicx}
\usepackage{amsmath,amssymb,amsfonts}%,amsthm}
\usepackage{epstopdf}

\usepackage{cite}
\usepackage{ctable}
\usepackage{enumerate}
\usepackage{tikz,pgfplots}
\usepackage{mathrsfs}
\usepackage{setspace}
\usepackage{algorithm, algorithmic}

\usepackage{multirow}
\normalsize

\def\figref#1{Figure~\ref{#1}}
\long\def\comment#1{}

\def\figref#1{Fig.~\ref{#1}}
\def\comment#1{}
\def\0{{\mathbf 0}}

\def\0{{\mathbf 0}}

\def\figref#1{Fig.~\ref{#1}}
\def\comment#1{}

\def\C(#1){|#1|}

\def\eqar#1{\begin{eqnarray} #1 \end{eqnarray}}
\def\eq#1{\begin{equation} #1 \end{equation}}
\def\eqn#1{\begin{equation} \nonumber #1 \end{equation}}

\def\Def{\triangleq}
\def\comment#1{\mbox{\;\;\;#1}}
\def\Ind(#1){\Delta\left[#1\right]}

\def\Inner(#1,#2){\langle #1, #2 \rangle}

\newtheorem{theorem}{{Theorem}}
\newtheorem{lemma}[theorem]{{Lemma}}

\newtheorem{corollary}[theorem]{{Corollary}}

\newtheorem{remark}{{Remark}}
\newcommand{\be}[1]{\begin{equation}\label{#1}}
\newcommand{\ee}{\end{equation}}
\newcommand{\deff}{\mbox{$\stackrel{\rm def}{=}$}}

\newcommand{\Pn}{F^{\otimes n}}
\newcommand{\script}[1]{{\mathscr #1}}

\newcommand{\sX}{\script{X}}
\newcommand{\sY}{\script{Y}}

\newcommand{\shalf}{\mbox{\raisebox{.8mm}{\footnotesize $\scriptstyle 1$}
\footnotesize$\!\!\! / \!\!\!$ \raisebox{-.8mm}{\footnotesize
$\scriptstyle 2$}}}
\newcommand{\cG}{{\cal G}}
\newcommand{\cB}{{\cal B}}
\newcommand{\cC}{{\cal C}}
\newcommand{\cT}{\mathcal{T}}
\newcommand{\Tref}[1]{Theo\-rem\,\ref{#1}}

\newcommand{\Cref}[1]{Co\-ro\-lla\-ry\,\ref{#1}}
\DeclareMathAlphabet{\mathbfsl}{OT1}{ppl}{b}{it} %{OT1}{cmr}{bx}{it}

\newcommand{\hatu}{\hat{u}}

\newlength\myindent
\newcommand\bindent{%
  \begingroup
  \setlength{\itemindent}{\myindent}
  \addtolength{\algorithmicindent}{\myindent}
}
\newcommand\eindent{\endgroup}
\begin{document}

\title{Relaxed Polar Codes }

\author{Mostafa El-Khamy, Hessam Mahdavifar, Gennady Feygin, Jungwon Lee, and Inyup Kang
\thanks{%
   This work was presented in part
   at the 2015 IEEE Wireless Communications and Networking Conference (WCNC), New Orleans, USA  \cite{ElkhamyRelaxed}.}
}

\maketitle

\begin{abstract}Polar codes are the latest breakthrough in coding theory, as they are the first family of codes with explicit construction that provably achieve the symmetric capacity of   discrete memoryless channels. Ar{\i}kan's polar encoder and successive cancellation decoder have  complexities of $N \log N$, for  code length $N$. Although, the complexity bound of $N \log N$ is asymptotically favorable, we report in this work methods to further reduce the encoding and decoding complexities of polar coding. The crux is to relax the polarization of certain bit-channels without performance degradation.
We consider schemes for relaxing the polarization of both \emph{very good} and \emph{very bad} bit-channels, in the process of channel polarization. Relaxed polar codes are proved to preserve the capacity achieving property of polar codes. Analytical bounds on the asymptotic and finite-length complexity reduction attainable by relaxed polarization are derived.
 For binary erasure channels, we show that the computation complexity can be reduced by a factor of 6, while preserving the rate and error performance. We also show that relaxed polar codes can be decoded with significantly reduced latency. For AWGN channels with medium code lengths, we show that relaxed polar codes can have lower error probabilities than conventional polar codes, while having reduced encoding and decoding computation complexities. \end{abstract}

\section{Introduction}

Polar codes, introduced by Ar{\i}kan \cite{arikan2009channel, arikan2009rate}, are the most recent breakthrough in coding theory. Polar codes are the first and, currently, the only family of codes with explicit construction (no ensemble to pick from) to asymptotically achieve the capacity of symmetric discrete memoryless channels as the block length goes to infinity. Besides their obvious application  in error correction, recent research have shown the possibility of applying polar codes and the polarization phenomenon in various communications and signal processing problems such as  data compression \cite{Arikan2,Ab}, BICM channels \cite{polarBICM}, wiretap channels \cite{MV}, multiple access channels \cite{STY,AT2,polar_MAC}, and broadcast channels \cite{polarBroadcast}.  There have also been various modified constructions of polar codes for the different applications, such as generalized polar codes \cite{KoradaPolarConst}, compound polar codes  \cite{compound}, concatenated polar codes \cite{mahdavifar2014concatenated}, and universal polar codes \cite{hassani2014universal}.

Polar codes can be encoded and decoded with relatively low complexity. Both the encoding complexity and the successive cancellation (SC) decoding complexity of polar codes are $O(N \log N)$, for code length $N$ \cite{arikan2009channel}. The decoding latency and memory requirements of polar decoders can be reduced to $O(N)$ \cite{Gross2, Gross1,MultiDimFast}. Hardware architectures for polar decoders, with $O(N)$ memory and processing elements, were implemented \cite{Gross2}. A semi-parallel architecture for SC decoding has been recently proposed \cite{Gross1}, where efficiency is achieved without a significant throughput penalty by sharing processing resources and taking advantage of the regular structure of polar codes.  The encoding and decoding latencies of polar codes can also be reduced to $O(N)$, through multi-dimensional polar transformations \cite{MultiDimFast}. Alamdar and Kschischang proposed  a simplified successive cancellation decoder with reduced latency and computational complexity by simplifying the decoder to decode all bits in a rate-one  or a rate-zero constituent code simultaneously \cite{alamdar2011simplified}. Reduction in decoding latency can also be achieved by changing the code construction, such as through interleaved concatenation of shorter polar codes \cite{mahdavifar2014concatenated}.

 In this paper, we propose methods to reduce both the encoding and decoding computational complexities of polar codes, by means of \emph{relaxing} the channel polarization. The resultant codes are called relaxed polar codes. Hence, hardware implementations for the encoders and decoders of relaxed polar codes can require smaller area and less power consumptions than conventional polar codes. Efficient methods for the implementation of the SC decoder, as in \cite{Gross2, Gross1,MultiDimFast}, can also be applied to  further improve the efficiency of decoding relaxed polar codes.

In practical scenarios, codes have finite block lengths and are designed with a specific information block length and rate in order to satisfy a certain error rate. Due to the nature of channel polarization, the error probability of certain bit channels decrease (or increase) exponentially at each polarization step. Hence, the encoding and decoding complexities can be reduced by relaxing the polarization of certain channels if their polarization degrees hit suitable thresholds, while satisfying the code rate and error rate requirements. For Ar{\i}kan's polar code with length $N$, each bit-channel is polarized $\log N$ times. However, for the proposed relaxed polar codes, some bit channels will be fully polarized $\log N$ times, and the polarization of the remaining bit-channels will be relaxed, where their polarization is aborted if they become sufficiently good or sufficiently bad with less than $\log N$ polarization steps.  Relaxed polarization results in fewer polarizing operations, and hence a reduction in complexity.  It is found that with careful construction of relaxed polar codes, there is no error performance degradation. In fact, it is observed that relaxed polar codes can have a lower error rate than conventional polar codes with the same rate.

The rest of this paper is organized as follows. In Section\,\ref{sec:sec1}, we give an overview of channel polarization theory and construction of conventional polar codes, which we call \emph{fully polarized} codes. In Section\,\ref{sec:sec2}, the notion of relaxed channel polarization is introduced and the relaxed channel polarization theory is established. The asymptotic bounds on the complexity reduction using relaxed polar codes are discussed in Section\,\ref{sec:sec3}. Then, upper bounds and lower bounds on the complexity reduction at finite block lengths are derived in Section\,\ref{sec:sec4}. These bounds are evaluated and compared with the actual complexity reductions at certain code parameters, in Section\,\ref{sec:sec5}.  Constructions of relaxed polar codes for general channels, and decoders for relaxed polar codes are discussed in Section\,\ref{sec:sec6}.  The relation between the relaxed polar code construction and the simplified successive cancellation decoder (SSCD) is discussed in Section \,\ref{sec:sec5b}. Numerical simulations on the AWGN channels are presented in Section\,\ref{sec:AWGN}. The paper is concluded in Section\,\ref{sec:sec7}.

%=======================================================================%
%                                                                       %
%     1. Ar{\i}kan's Fully Polarized Codes                                   %
%                                                                       %
%=======================================================================%

\section{Ar{\i}kan's Fully Polarized Codes}
\label{sec:sec1}

For any binary-input discrete memoryless channel (B-DMC) $W : \sX \rightarrow \sY$, let $W(y|x)$ denote the probability of receiving $y \in \!\sY$ given that $x \in \!\sX = \{0,1\}$ was sent, for any $x \in \!\sX$ and $y \in \!\sY$.
For an B-DMC $W$, the \emph{Bhattacharyya parameter} of $W$ is
\eq{
Z(W)
\,\ \deff\kern1pt
\sum_{y\in\sY} \!\sqrt{W(y|0)W(y|1)}.
}
The symmetric capacity of a B-DMC $W$ can be written as
\eq{I(W) \Def  \sum_{y\in\sY} \frac{1}{2} \sum_{x\in\sX} W(y|x) \log \frac{W(y|x)}
{\frac{1}{2} W(y|0) + \frac{1}{2} W(y|1)}.}

For a binary memoryless symmetric (BMS) channel with uniform input, the error probability of $W$ can be characterized as
\eq{ E(W) = \frac{1}{2} \sum_{y\in\sY} \min\{W(y|0), W(y|1)\}.}
The Bhattacharyya parameter $Z(W)$ can be shown to be always between $0$ and $1$. The Bhattacharyya parameter can be regarded as a measure of the reliability of $W$. Channels with $Z(W)$ close to zero are almost noiseless, while channels with $Z(W)$ close to one are almost pure-noise channels. More precisely, it can be proved that the probability of error of an BMS channel is upper-bounded by its Bhattacharyya parameter \cite{BPKailath}
\eq{\label{BP-Err} 0 \leq 2E(W)  \leq Z(W) \leq 1. }

The construction of polar codes is based on a phenomenon called \emph{channel polarization} discovered by Ar{\i}kan \cite{arikan2009channel}. Consider the polarization matrix
\eq{\label{G-def}
F
\ = \
\left[
\begin{array}{c@{\hspace{1.25ex}}c}
1 & 0\\
1 & 1\\
\end{array}
\right]
. }

Consider the $2 \times 2$ polarizing transformation $F$ which takes two independent copies of $W$ and performs the mapping $(W,W) \rightarrow (W^{-}, W^{+})$, where
$W: \{0,1\} \to \sY$, $W^-: \{0,1\} \to \sY^2$, and  $W^+: \{0,1\} \to \{0,1\} \times \sY^2$, then polarization is defined with the channel transformation
\eqar{W^-(y_1,y_2|x_1) = \frac{1}{2}  \sum_{x_2 \in \{0,1\}} W(y_1 | x_1 \oplus x_2) W(y_2|x_2), \\
W^+(y_1,y_2,x_1|x_2) =  \frac{1}{2} W(y_1 | x_1 \oplus x_2) W(y_2|x_2),
}
where $W^-$ and $W^+$ are degraded and upgraded channels respectively.
Hence, the following is true for the bit-channel rates
 \cite{arikan2009channel},
\eqar{ I(W^+) + I(W^-) &=& 2 I(W) \\ I(W^-) \leq  I(W) & \leq & I(W^+).}

The mapping of $(W,W) \rightarrow (W^{-}, W^{+})$ is called one level of polarization. The same mapping is applied to $W^{-}$ and $W^{+}$ to get $W^{--}$, $W^{-+}$, $W^{+-}$, $W^{++}$, which is the second level of polarization of $W$. The same process can be continued in order to polarize $W$ for any arbitrary number of levels. The polarization process can be also described using a binary tree, where the root of the tree is associated with the channel $W$. Each node in the binary tree is associated with some bit-channel $W'$ and has two children, where the left child corresponds to $W^{'-}$ and the right child corresponds to $W^{'+}$.

The channel polarization process can be also represented using the Kronecker powers of $F$ defined as follows. $F^{\otimes 1} = F$ and for any $i > 1$,
$$
F^{\otimes i}
\ = \
\left[
\begin{array}{c@{\hspace{1.25ex}}c}
F^{\otimes (i-1)} & 0\\
F^{\otimes (i-1)} & F^{\otimes (i-1)}\\
\end{array}
\right],
$$
where $F^{\otimes i}$ is a $2^i \times 2^i$ matrix. Let $n = \log_2 N$. Then, the $N \times N$ polarization transform matrix is defined as
$G_N \Def B_N F^{\otimes n}$,
where $B_N$ is the bit-reversal permutation matrix \cite{arikan2009channel}.
Let $u_1^N$ denote the vector $(u_1, u_2,\dots,u_N)$ of $N$ independent and uniform binary random variables. Let $x_1^N = u_1^N G_N$ be transmitted through $N$ independent copies of a binary-input discrete memoryless channel (B-DMC) $W$ to form channel output $y_1^N$.
Let  $W^N: \sX^N \rightarrow \sY^N$ denote the channel that results from $N$ independent copies of $W$ in the polar transformation i.e.
$W^N\kern-0.5pt(y^N_1|x^N_1)
\,\ \Def\,\
\prod_{i=1}^N W(y_i|x_i)
$.
The combined channel ${W}_N$ is defined
with transition probabilities given by
\be{Wtilde}
{W}_N(y^N_1|u^N_1)
\,\ \Def\,\
W^N\kern-1pt\bigl(y^N_1\hspace{1pt}{\bigm|}\hspace{1pt}u^N_1\hspace{1pt} G_N\bigr)
\kern1pt = \kern2pt
W^N\kern-1pt
\bigl(y^N_1\hspace{1pt}{\bigm|}\hspace{1pt}u^N_1\hspace{1pt} B_N\Pn \bigr).
\ee
This is the channel that the random vector $u_1^N$
 observes through the polar transformation.

 Assuming uniform channel input and a genie-aided successive cancellation decoder, the bit-channel $W^{(i)}_N$ is defined with the following transition probability:
\be{Wi-def}
W^{(i)}_N\bigl( y^N_1,u^{i-1}_1 | \hspace{1pt}u_i)
\,\ \Def \,\
\frac{1}{2^{N-1}}\hspace{-5pt}
\sum_{u_{i+1}^N \in \{0,1\}^{N-i}} \hspace{-12pt}
{W}_N\Bigl(y^N_1\hspace{1pt}{\bigm|}\hspace{1pt}
u_1^N \Bigr).
\ee
Notice that $W^{(i)}_N$ gives the transition probabilities of $u_i$ assuming all the preceding bits $u_1^{i-1}$ are already decoded and are available, together with the $N$ observations at the channel output $y_1^N$. This is actually the channel that $u_i$ observes and is also referred to as the $i$-th bit-channel. It can be observed that $W^{(i)}_N$ corresponds to the $i$-node in the $n$-th level of polarization of $W$. The following recursive formulas hold for Bhattacharyya parameters of individual bit-channels in the polar transformation \cite{arikan2009channel}
\begin{align}
\label{Z_recursion1}
Z(W_{2N}^{(2i-1)}) &\leq 2Z(W_N^{(i)}) - Z(W_N^{(i)})^2\\
\label{Z_recursion2}
Z(W_{2N}^{(2i)}) &= Z(W_N^{(i)})^2
\end{align}
with equality iff $W$ is a binary erasure channel.

 The channel polarization theorem states that as the code length $N$ goes to infinity, the bit-channels become polarized, meaning that they either become noise-free or very noisy.
Define the set of good bit-channels according to the channel $W$ and a positive constant $\beta \,{<}\, \shalf$ as
\begin{eqnarray}
\label{good-def}
\cG_N(W,\beta)
&\hspace*{-6pt}{\deff}\hspace*{-6pt}&
\left\{\, i \in [N] ~:~ Z(W^{(i)}_N) < 2^{-N^{\beta}}\!\!/N \hspace{1pt}\right\},
\end{eqnarray}
where $[N] \ \ \Def \ \{1,2,\dots,N\}$, then the main channel polarization theorem follows ~\cite{arikan2009channel,arikan2009rate}:

\begin{theorem}
\label{polar_thm1}
 For any BMS channel $W$, with capacity $\cC(W)$,  and any constant $\beta \,{<}\, \shalf$,
$$\lim_{N \to \infty} \frac{\left|\cG_N(W,\beta)\right|}{N}
\,=\,
\cC(W).$$
\end{theorem}
\Tref{polar_thm1} readily leads to a construction of capacity-achieving \emph{polar codes}. The crux of polar codes is to carry the information bits on the upgraded noise-free channels and freeze the degraded noisy channels to a predetermined value, e.g. zero. The following theorem shows the error exponent under successive cancellation decoding \cite{arikan2009channel}:

\begin{theorem}
\label{polar_thm2}
Let $W$ be a BMS channel and let $k=\left|\cG_N(W,\beta)\right|$ be the cardinality of the information bits, which are encoded using a polar code of length $N$, and transmitted
over $W^N$, then the probability of decoder error under successive
cancellation decoding satisfies
$ P_e \, \leq \,
\sum_{i \in \cG_N(W,\beta)} \!Z(W^{(i)}_N) \leq 2^{-N^{\beta}}.
$
\end{theorem}

Similar to the set of good bit-channels, the set of bad bit-channels is defined according to the channel $W$ and a positive constant $\beta \,{<}\, \shalf$ as
\begin{eqnarray}
\label{bad-def}
\cB_N(W,\beta)
&\hspace*{-6pt}{\deff}\hspace*{-6pt}&
\left\{\, i \in [N] ~:~ Z(W^{(i)}_N) > 1 - 2^{-N^{\beta}} \hspace{1pt}\right\},
\end{eqnarray}

The following corollary can be derived by specializing the Theorem\,3 of ~\cite{hassani2010scaling}:
\begin{corollary}
\label{cor-hassani}
Let $W$ be an arbitrary BSM channel. Then, for
any positive constant
$\beta < \shalf$,
\be{cor-limit}
\lim_{N \rightarrow \infty} \frac{\left|\cB_N(W, \beta)\right|}{N}
\:=\,
1-\cC(W).
\ee
\end{corollary}

%=======================================================================%
%                                                                       %
%     2. Relaxed Polarization Theory                                    %
%                                                                       %
%=======================================================================%

\section{Relaxed Polarization Theory}
\label{sec:sec2}

In this section, we define relaxed polarization. We prove that, similar to conventional polar codes, relaxed polar codes can asymptotically achieve the capacity of a binary memoryless symmetric channel. We also prove that the bit-channel error probability of relaxed polar codes is at most that of conventional polar codes without rate-loss.

Let us observe the definition of good channels in \Tref{polar_thm1}. Let us also observe that the Bhattacharrya parameters (BP) approach $0$ or $1$ exponentially with the block length $N$. Let $\widetilde{W}$ denote the bit-channels of the relaxed polar code. Consider two independent copies of a parent bit-channel $i$ at polarization level $t$ to be polarized into two bit-channel children at level $t+1$, corresponding to a code of length $2^{t+1}$, via the following channel transformation
\eq{\left(\widetilde{W}_{2^t}^{(i)}, \widetilde{W}_{2^t}^{(i)} \right) \to  \left(\widetilde{W}_{2^{t+1}}^{(2i-1)}, \widetilde{W}_{2^{t+1}}^{(2i)} \right).}

Consider the case, when the polarized channel $\widetilde{W}_{2^t}^{(i)}$ at level $t < n$, where $n = \log N$, is sufficiently good, such that it satisfies the definition of a good channel at the target length $2^n$, i.e. $Z(\widetilde{W}^{(i)}_{2^t}) < 2^{-N^{\beta}}/N$. Then, the idea of relaxed polarization is to stop further polarization of this good channel, and the corresponding node in the polarization tree is called a relaxed node, such that the channels of all the descendents of a relaxed node are the same as that of the relaxed parent node and will also be relaxed. Let $u^j_{1,o}$ and $u^j_{1,e}$ denote the sub-vectors with odd and even indices, respectively.
Then, the bit-channel transformation at the relaxed node is given by
\begin{align*}
\widetilde{W}_{2^{t+1}}^{(2i-1)}(y_1^{2^{t+1}}, u_1^{2i-2} | u_{2i-1}) &=& \widetilde{W}_{2^{t}}^{(i)}(y_1^{2^{t}}, u_{1,o}^{2i-2}|u_{2i-1}) \\
 \widetilde{W}_{2^{t+1}}^{(2i)}(y_1^{2^{t+1}}, u_1^{2i-1} | u_{2i}) &=& \widetilde{W}_{2^{t}}^{(i)}(y_{2^{t}+1}^{2^{t+1}}, u_{1,e}^{2i-2}|u_{2i}).
\end{align*}

Relaxing the further polarization of sufficiently good channels is called good-channel relaxed polarization.
For the good-channel relaxed polar code, define the set of good bit-channels according to the channel $W$ and a positive constant $\beta \,{<}\, \shalf$ as
\begin{eqnarray}
\label{good-def-rlx}
\widetilde{\cG}_N(\widetilde{W},\beta)
&\hspace*{-6pt}{\deff}\hspace*{-6pt}&
\left\{\, i \in [N] ~:~ Z(\widetilde{W}^{(i)}_N) < 2^{-N^{\beta}}\!\!/N \hspace{1pt}\right\}.
\end{eqnarray}
 Next, we show that relaxed polar codes, similar to fully polarized codes, asymptotically achieve the capacity of BMS channels.
\begin{theorem}
\label{polar_thm3}
{ For any BMS channel $W$, with capacity $\cC(W)$,  and any constant $\beta \,{<}\, \shalf$,}
$$
\lim_{N \to \infty} \frac{\left|\widetilde{\cG}_N(\widetilde{W},\beta)\right|}{N}
\,=\,
\cC(W)\vspace{1.5ex}.
$$
\end{theorem}
\begin{proof} Consider a relaxed channel at level $t < n$, where $n=\log N$. Then $Z(\widetilde{W}^{(i)}_{2^t}) < 2^{-N^{\beta}}/N$. Then the BPs of all its $2^{n-t}$ descendents at level $n$  are equal to
$Z(\widetilde{W}^{(i)}_{2^t})$ and are in  $\widetilde{\cG}_N(\widetilde{W},\beta)$.
In case of full polarization, if a channel belongs to ${\cG}(W,\beta)$, then it must have polarized to a good channel at level $n$ or earlier. If it polarized at level $n$, then by definition it also belongs to $\widetilde{\cG}_N(\widetilde{W},\beta)$. Otherwise, its parent has polarized to a good channel at level $t<n$. With relaxed polarization, this channel and all its $2^{n-t}-1$ siblings will also be in  $\widetilde{\cG}_N(\widetilde{W},\beta)$. Therefore, ${\cG}(W,\beta) \subset \widetilde{\cG}_N(\widetilde{W},\beta)$ and hence $|\widetilde{\cG}_N(\widetilde{W},\beta)| \geq |{\cG}(W,\beta)|$. Then, the proof follows from \Tref{polar_thm1}.
\end{proof}

The upper bound $2^{-N^{\beta}}$ on the probability of error as in \Tref{polar_thm2} is still valid for the relaxed polar code constructed  with respect to $\widetilde{\cG}_N(\widetilde{W},\beta)$. Hence, \Tref{polar_thm3} shows that it is possible to construct  good-channel relaxed polar codes, which are still capacity achieving.

The remaining question is to actually compare the bit-error rate of relaxed polar code with that of Ar{\i}kan's polar code. Consider the special case when $\widetilde{\cG}_N(\widetilde{W},\beta) = \cG_N({W},\beta)$. Consider the channel ${W}^{(i)}_{N/2}$, then we have the following inequalities (the proof is provided in the Appendix.)
\eqar{  E\left(W_{N}^{(2i-1)}\right) &=&  2 E\left(W_{N/2}^{(i)}\right)  -  2 E\left(W_{N/2}^{(i)}\right)^2 \label{eq:Eineq1}\\
\label{eq:Eineq2}
 E\left(W_{N}^{(2i)}\right) &\geq& 2 E\left(W_{N/2}^{(i)}\right)^2  }
Consider a good-channel relaxed polar code with good-channel set  $\widetilde{\cG}_N(\widetilde{W},\beta)$, which is assumed to be equal to the good channel set of the fully polarized polar code, i.e., $\widetilde{\cG}_N(\widetilde{W},\beta) = \cG_N({W},\beta)$. Consider the last level of channel polarization e.g. channel ${W}^{(i)}_{N/2}$ and its children, $W_{N}^{(2i-1)}$ and $W_{N}^{(2i)}$, assuming that  ${W}^{(i)}_{N/2}$ is a relaxed node. Then, both indices $2i-1$ and $2i$ are contained in $\widetilde{\cG}_N(\widetilde{W},\beta)$ and $\cG_N({W},\beta)$. For the relaxed code, it follows that sum error probability of these two channels is $
E\left(\widetilde{W}_{N}^{(2i-1)}\right) +
E\left(\widetilde{W}_{N}^{(2i)}\right) = 2 E\left(\widetilde{W}^{(i)}_{N/2} \right) =
2 E\left({W}^{(i)}_{N/2} \right)$. Together with summing \eqref{eq:Eineq1} and \eqref{eq:Eineq2}, it follows that $2 E\left(\widetilde{W}^{(i)}_{N/2} \right) \leq E\left(W_{N}^{(2i-1)}\right) + E\left(W_{N}^{(2i)}\right)$.
Therefore, we have the following lemma.
\begin{lemma} \label{lem:error}
Let a good-channel relaxed polar code have a good-channel set  $\widetilde{\cG}_N(\widetilde{W},\beta)$, which is equal to the good channel set of the fully polarized polar code, i.e., $\widetilde{\cG}_N(\widetilde{W},\beta) = \cG_N({W},\beta)$, then
\begin{equation}
\label{eq:FER}
\sum_{i \in \widetilde{\cG}_N(\widetilde{W},\beta)} \! E\left(\widetilde{W}^{(i)}_N \right) \leq
\sum_{i \in \cG_N({W},\beta)} \!E\left({W}^{(i)}_N \right).
\end{equation}
\end{lemma}
Note that the left hand side of \eqref{eq:FER} is the union bound on the frame error probability (FER) of the constructed relaxed polar code while the bound is very tight for low FERs. Similarly, the right hand side of \eqref{eq:FER} is the union bound on the frame error probability (FER) of the constructed relaxed polar code which is also very tight for low FERs. Hence, we conclude that the relaxed polar code is expected to perform better than the fully polarized codes in terms of frame error rate. This will be verified in Section\,\ref{sec:AWGN}.

\begin{remark}
The concept of good-channel relaxed polarization, discussed so far, can be extended to \emph{bad-channel relaxed polarization} as follows. Consider the bit-channels in the polarization tree, at a level $t <n$, that are \emph{very bad}. A bit-channel can be considered very bad if its Bhattacharyya parameter is very close to $1$, or if none of its descendents will fall into the set of good bit-channels at the last level of polarization $n$. Hence, the polarization of very bad bit-channels can be stopped without affecting the final set of good bit-channels. Thus, with careful bad-channel relaxed polarization, more complexity reductions can be possible, without degrading the code rate and error performance.
\end{remark}
\begin{remark} The obvious advantage of relaxed polarization is savings in both encoding and decoding complexities, since there will be no channel transformations done at relaxed nodes while encoding, and there will be no need to calculate new likelihood ratios (LRs) at relaxed nodes while decoding. Hence, relaxed polarization will result in a reduction in the encoding and decoding computational complexities of polar codes. Another advantage of relaxed polarization is the reduction in space (area and memory) complexity in practical implementations. That is because decoding of relaxed polar codes will result in smaller LRs (or log-likelihood ratios (LLRs) in case the computation is done in the log domain \cite{Gross2}) than that for fully polarized polar codes, and hence smaller bit-width will be required for LR calculation and storage. Relaxed polarization has the effect of reducing the number of processing nodes required at lower levels of the polarization tree, and hence one can expect even more efficient implementations (or less throughput penalty) with semi-parallel hardware architectures \cite{Gross1}. Also, by appropriate permutation of the bit channels, one expects to be able to eliminate the wiring for the wider butterflies in FFT-like SC decoders for relaxed polar codes.  The reduction in bit-width and number of processing nodes required for relaxed polar code decoders has the compound effect of reduction in power consumption.
\end{remark}
The reduction in encoding and decoding complexity will be addressed in the following section.

%=======================================================================%
%                                                                       %
%     3. Asymptotic Analysis of Complexity Reduction                    %
%                                                                       %
%=======================================================================%

\section{Asymptotic Analysis of Complexity Reduction}
\label{sec:sec3}

In this section, we establish bounds on the asymptotic  complexity reduction (as the code's block length goes to infinity) in polar code encoders and decoders, made possible by relaxed polarization.

First, we elaborate the notion of complexity reduction. For Ar{\i}kan's polar codes, the total number of channel polarization operations required using Ar{\i}kan's butterfly polarization structure, is $$A(n) = nN,$$ where $N = 2^n$ is the length of the code. As a result, the encoding procedure consists of $nN$ binary XOR operations and decoding procedure consists of $nN$ LLR combinations. Therefore, each skipped polarization operation is equivalent to one \emph{unit} of complexity reduction in both encoding and decoding, where the unit corresponds to a binary XOR when encoding and an LLR combining operation when decoding.
The complexity reduction $\mathcal{R}(n)$ is defined to be the ratio of the number of polarization operations that are skipped due to relaxed polarization  to the total number of polarization operations, $A(n)$, required for full polarization. The complexity reduction (CR) can be directly translated into encoding and decoding complexity ratios of  $(1-\mathcal{R}(n))^{-1}$.

For the asymptotic analysis throughout this section, a family of capacity-achieving polar codes is assumed which is constructed with respect to a fixed parameter $\beta\, <\, \shalf$ and the set of good bit-channels $\cG_N({W},\beta)$, for any block length $N=2^n$.

\begin{theorem}
\label{asymp-bound1}
Let $C(W)$ be the capacity of the channel $W$. Then, for any $\epsilon > 0$, small enough $\delta > 0$, and large enough $N$, the complexity reduction ratio using the relaxed polar code, constructed with $\widetilde{\cG}_N(\widetilde{W},\beta)$, is at least
$$
(C - \epsilon)\bigl(1 - (2+\delta)\beta \bigr)
$$
\end{theorem}
\begin{proof}
Pick a fixed $\delta$ such that $0< \delta < 1/\beta - 2$. Consider the polarization level $\left\lceil(2+\delta)\beta n \right\rceil$. Let $N' = 2 ^{\left\lceil(2+\delta)\beta n \right\rceil}$ be the total number of nodes at this level. Then for large enough $n$, the nodes at this level with index belonging to the set $\cG_{N'}({W},1/(2+\delta))$ will be relaxed. Notice that the fraction of these nodes, i.e. $\left|\cG_{N'}({W},1/(2+\delta))\right| / N'$, approaches the capacity $C$ by \Tref{polar_thm1}. The fraction of bit-channel polarizations between the level $\left\lceil(2+\delta)\beta n \right\rceil$ and the last level $n$ is $(1 - (2+\delta)\beta \bigr)$ of the total $nN$, among which a fraction of $C - \epsilon$ of them are relaxed, for large enough $n$. Therefore, the complexity reduction will be at least $(C - \epsilon)\bigl(1 - (2+\delta)\beta \bigr)$.
\end{proof}

In the next theorem, a bound on the asymptotic complexity reduction using the bad-channel relaxed polarization is provided. The following scenario is considered for bad-channel relaxed polarization: if none of the descendants of a certain node will belong to $\cG_N(W, \beta)$, then the polarization at that node, and consequently all of its descendants, will be relaxed.

\begin{theorem}
\label{asymp-bound2}
For any $\epsilon, \delta > 0$ and large enough $N = 2^n$, the complexity reduction ratio using the bad-channel relaxed polarization is at least
$$
(1-C-\epsilon)\left(1 - \frac{(2+\delta) \log n}{n}\right)
$$
\end{theorem}
\begin{proof}
Consider the polarization level $t = \left\lceil (2+\delta) \log n \right\rceil$. Then by \Cref{cor-hassani}, for large enough $n$, the fraction of nodes with Bhattacharyya parameter at least $1 - 2^{-2^{t/(2+\delta)}} \geq 1 - 2^{-n}$ is at least $1 - C - \epsilon$. Consider such a node $V$ with Bhattacharyya parameter $Z \geq 1 - 2 ^{-n}$. Then the best descendant of $V$ at the last level of polarization has Bhattacharyya parameter
$$
Z ^ {2^{n-t}} >  Z^{2^n} > (1 - 2^{-n})^{2^n} > \frac{1}{e},
$$
which implies that it can not be a good-bit-channel. Therefore, the total fraction of complexity reduction is at least
$$
(1-C-\epsilon)\frac{n-t}{n},
$$
and the theorem follows.
\end{proof}

Observe that, by neglecting $\epsilon$ and $\delta$ in the bounds given in \Tref{asymp-bound1} and \Tref{asymp-bound2} and by assuming large enough $n$, the complexity reduction ratio from good-channel and bad-channel relaxed polarization is $C(1-2\beta)$ and $1-C$, respectively, which are both positive constant factors.  By combining both good and bad channel relaxation, the ratio of saved operations approaches $1 - 2 \beta C$. Hence, relaxed polarization can provide a non-vanishing scalar reduction in complexity, even as the code length grows infinitely.

%=======================================================================%
%                                                                       %
%     4. Finite Length Analysis of Complexity Reduction                 %
%                                                                       %
%=======================================================================%

\section{Finite Length Analysis of Complexity Reduction}
\label{sec:sec4}

In this section, we derive bounds on the complexity reduction resulted from good-channel and bad-channel relaxed polarization at finite block lengths.
\subsection{Relaxed polar code constructions using Bhattacharyya parameters \label{subsec: BP}}
In general, finite block length polar codes are constructed by fixing either a target frame error probability (FER) $E$ or target code rate $R$. We consider construction of polar codes with code length $N=2^n$, at a target FER of $E$.
To simplify notation, let $Z_{i,t} \Def Z(\widetilde{W}^{(i)}_{2^t})$. At finite block lengths, we need to specify certain thresholds for Bhattacharyya parameters in order to establish criteria for good-channel and bad-channel relaxed polarization. As a result, the following scenarios are considered for relaxed polarization:

 \begin{enumerate}
 \item Good-Channel Relaxed Polarization (GC-RP):\\  Node $i$ at polarization level $t$ is not further polarized if
 $Z_{i,t} <\mathcal{T}_g$
 \item Bad-Channel Relaxed Polarization (BC-RP): \\
 Node $i$ at polarization level $t$ is not further polarized if
 $Z_{i,t} >\mathcal{T}_b $
 \item All-Channel Relaxed Polarization (AC-RP): \\
 Node $i$ at polarization level $t$ is not further polarized if
 $Z_{i,t} <\mathcal{T}_g$ or $Z_{i,t} >\mathcal{T}_b$
  \end{enumerate}
where $\mathcal{T}_g$ and $\mathcal{T}_b$ are thresholds that can be considered as parameters of the construction.

\begin{remark} If $\mathcal{T}_g = 2E/N$, then GC-RP constructed codes satisfy the target FER $E$. Let $\Gamma$ be the set of \emph{good} bit-channel channels which are used to transmit the information bits.
Then, this can be observed by \eqref{BP-Err} and the fact that $|\Gamma| \leq N$, which gives the following
\eq{\mbox{FER} \leq \sum_{i \in \Gamma}  Z_{i,n}/2 \leq | \Gamma|\, E/N \leq E.}
\label{rem:goodT}
\end{remark}
In the proposed bad-channel relaxed polarization (BC-RP), the bad channels are not further polarized if they become sufficiently bad, where the bad-channel relaxation threshold can be set to be $\mathcal{T}_b = 1-\mathcal{T}_g$. To guarantee no rate loss from BC-RP, it should only be performed if $n-t \leq \lceil \log_2 \frac{\log_ 2 \mathcal{T}_g}{\log_2 \mathcal{T}_b} \rceil$.
The proposed all-channel relaxed polarization (AC-RP) relaxes the polarization of a bit-channel if it becomes either sufficiently good or bad. Since the bad bit-channels do not contribute to the FER, the target FER is still maintained with AC-RP.

 \begin{figure}[t!]
\centering
\includegraphics[width=\linewidth]{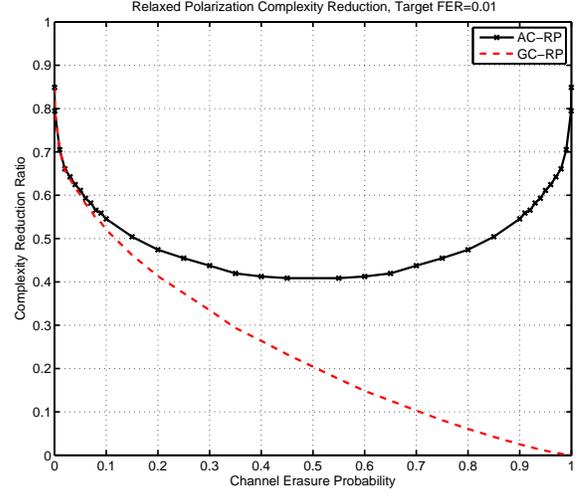}%{relax_CR_L1000em5.eps} %{BLCC.eps}\\%{LM2SMLM.eps}\\
\caption{Complexity reduction via relaxed polarization on erasure channels at $n=20$}
\label{fig:CR_erasure}
\end{figure}

In \figref{fig:CR_erasure}, the achieved complexity reduction ratio for a binary erasure channel with erasure probability $p$, BEC($p$), is shown. It is observed that up to 85$\%$ CR is achievable, i.e. fully polarized (FP) code requires 6.6-fold the complexity of RP code. AC-RP will result in more complexity reduction than GC-RP  as the channel becomes worse. The rate-loss is calculated as
 \eq{\label{eq:RL} R_{Loss} = R_{FP} - R_{RP}, }
 where $R_{FP}$ and $R_{RP}$ are the rates of the codes which are constructed by full and relaxed polarization, respectively. The rate at a target FER $E$ is calculated by aggregating the maximum number of bit-channels such that their accumulated BPs does not exceed  $2E$. In this simulation result shown in \figref{fig:CR_erasure}, the rate loss is always less than $10^{-4}$.
Another important observation, from \figref{fig:CR_erasure}, is the symmetry of the CR curve around $p=0.5$. This is explained by the following Theorem \ref{Th:duality}, which is a direct result of Lemma \ref{Lem:mapping} and the description of GC-RP and BC-RP.

 \begin{figure}[t!]
\centering
\includegraphics[width=0.8\linewidth]{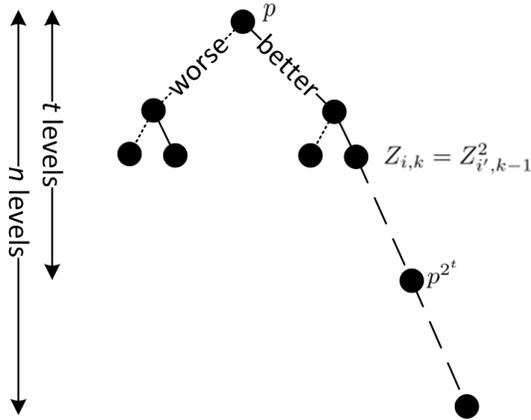} %{BLCC.eps}\\%{LM2SMLM.eps}\\
\caption{Upper bound (UB) on Complexity Reduction.}
\label{fig:P1}
%\vspace{-0.5cm}
\end{figure}

\begin{theorem} \label{Th:duality}
The complexity reduction with bad-channel relaxed polarization  of BEC($p$) with threshold $1 - \cT$ is the same as that of good-channel relaxed polarization of BEC($1-p$) with threshold $\cT$.
\end{theorem}

\begin{lemma} \label{Lem:mapping}
\label{lem:map}
Let the nodes in the polarization tree be labeled by their BPs. Then, the  polarization tree for BEC($p$) and the polarization tree for BEC($1-p$) are isomorphic, where a node $V$ with BP $Z$ in the first tree is isomorphic to a node with BP $1-Z$ in the second tree.
\end{lemma}
\begin{proof}
We show that the one-to-one mapping is nothing but mirroring i.e. the $i$-th node at the polarization level $t$ will be mapped to the node indexed by $2^t - i$ at the same level. It is sufficient to show this for one polarization level and then the rest follows from induction. Let $W_1 = \text{BEC}(p)$ and $W_2 = \text{BEC}(1-p)$. Then
$$
Z(W_2^{+}) = (1 - p)^2= 1 - (2p - p^2) = 1 - Z(W_1^{-})
$$
And
$$
Z(W_2^{-}) = 2(1-p) - (1-p)^2 = 1 - p^2 = 1 - Z(W_1^{+})
$$
Therefore, by induction on the polarization level, it is shown that each polarized node $V$ in the polarization tree of $W_1$ can be mapped to the polarization tree of $W_2$ by reversing the sequence of $+$'s and $-$'s during its polarization. Furthermore, the BP $Z$ of $V$ will be mapped to BP $1-Z$ at the image of $V$.
\end{proof}

\subsection{Analysis of complexity reduction for GC-RP}

In this subsection, bounds on the complexity reduction from good-channel relaxed polarization are discussed.
Let $\mathcal{T} = \mathcal{T}_g$. In the next theorem, a simple upper bound is provided, which is also illustrated in \figref{fig:P1}.

\begin{theorem}\label{Th:GCRPUB} The good-channel relaxed polarization complexity reduction is upper bounded by
\eq{\mathcal{R}_g(n) \leq (n-t_g)/n,}
where $t_g = \left \lceil \log_2 \frac{\log_2 \mathcal{T}}{\log_2 p}\right \rceil$ and $p$ is the erasure channel parameter. \end{theorem}

\begin{proof} The upper bound follows by considering the minimum number of polarization levels required for the best polarized channel to reach the threshold $\mathcal{T}$. Notice that $Z_{1,0} = p$. Then, after $t$ polarization levels, the minimum BP among all $Z_{i,t}$ is indeed $Z_{2^t,t} = p^{2^{t}}$. Hence, $t_g$ polarization levels are required for the BP of at least one bit channel to be less than $\mathcal{T}$. The upper bound on saved operations follows by skipping all polarization steps at all remaining $n-t_g$ levels. \end{proof}

Next, we derive lower bounds on the complexity reduction with relaxed polarization for $BEC(p)$. For any polarization level $t$ and some threshold $B$, let $\mathcal{G}_{B,t}$ denote the set of bit channels at polarization level $t$ with BP at most $B$ i.e. $\mathcal{G}_{B,t} = \{i: Z_{i,t} \leq B \}$.

 \begin{figure}[t!]
\centering
\includegraphics[width=0.8\linewidth]{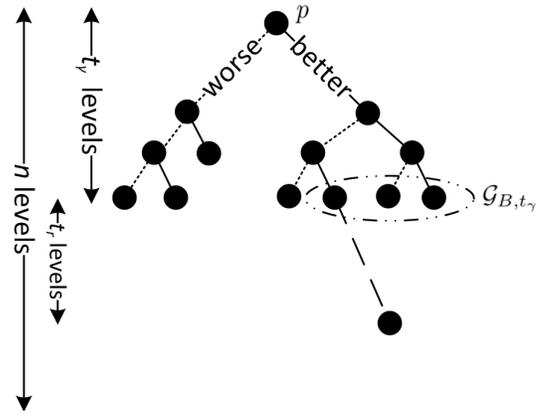} 
\caption{Lower bound (LB1) on Complexity Reduction.}
\label{fig:LB1}
\end{figure}

\begin{theorem} \label{Th:GCLB1} Let $t_{b} = \left \lceil \log_2  \frac{\log_2 B}{\log_2 p} \right \rceil$, for an arbitrary threshold $B$. For a polarization level $t_{\gamma} \geq t_b$, let $\gamma = |\mathcal{G}_{B,t_{\gamma}}|/2^{t_{\gamma}}$. Then the good-channel relaxed polarization complexity reduction is lower bounded by
\eqn{\mathcal{R}_g(n)  \geq  \gamma 2^{-t_r}(n-t_r-t_{\gamma})/n}
where $t_r= \left \lceil \log_2 \frac{\log_2 \mathcal{T}}{\log_2{B}} \right \rceil$. \end{theorem}

\begin{proof}
Notice that $t_{\gamma} \geq t_b$ guarantees that $\mathcal{G}_{B,t_{\gamma}}$ is a non-empty set. In polarization level $t_r+t_{\gamma}$, any node in $\mathcal{G}_{B,t_{\gamma}}$ has at least one descendant with BP less than $\mathcal{T}$, i.e. the right-most descendant which has BP at most $B^{2^{t_r}} \leq \mathcal{T}$. Therefore, there are at least $\gamma 2^{t_{\gamma}} = \left|\mathcal{G}_{B,t_{\gamma}}\right|$ nodes at polarization level $t_r+t_{\gamma}$ that have BP less than $\mathcal{T}$, and will be relaxed. Relaxing each of these nodes is equivalent to skipping $S=(n-t_r -t_{\gamma})2^{n-t_r -t_{\gamma}}$ polarization steps. Then the total number of polarization steps skipped is
$\gamma 2^{t_{\gamma}} S$, and the proof follows.
\end{proof}

\begin{corollary}
Let $t_g= \left \lceil \log_2 \frac{\log_2 \mathcal{T}}{\log_2{p}} \right \rceil$. Then the good-channel relaxed polarization complexity reduction is lower bounded by
\eqn{\mathcal{R}_g(n)  \geq  2^{-t_g}(n-t_g)/n}
\end{corollary}
\begin{proof}
The corollary follows by taking $B = p$, and $t_{\gamma} = t_b = 0$ in \Tref{Th:GCLB1}.
\end{proof}
The following provides a tighter lower bound on the GC-RP complexity reduction, which is also illustrated in Figure \ref{fig:LB2}.

\begin{theorem}\label{Th:GCLB2}  Let $t_{b} = \left \lceil \log_2  \frac{\log_2 B - 1}{\log_2 p} \right \rceil + 1$, for an arbitrary threshold $B$. For a polarization level $t \geq t_b$, let $\gamma' = \min_{t \geq t_b} |\left\{i\ \text{odd}: i \in \mathcal{G}_{B,t}\right\}|/2^t$.  Then the good-channel relaxed polarization complexity reduction is lower bounded by
\eqn{\mathcal{R}_g(n)  \geq  \frac{\gamma' 2^{-t_r}}{2n} \left( (n-t_r)(n-t_r-2t + 1) + t(t-1) \right)}
where $t_r = \left \lceil \log_2 \frac{\log_2 \mathcal{T}}{\log_2 B} \right \rceil$.
\end{theorem}

\begin{proof}
Consider the right-most node at polarization level $t - 1$ which has BP $p^{2^{t - 1}} \leq p^{2^{t_b - 1}} \leq B/2$. Therefore, the left child of this node is contained in $\mathcal{G}_{B,t}$ which means that the set of odd-indexed nodes in $\Gamma_{B,t}$ is always non-empty for $t \geq t_b$. The right-most descendant of any of these nodes, after $t_r$ more polarization levels, will have BP less than $\cT$ and will be relaxed by eliminating the polarizing subtrees emanating from them. The total reduction of polarization steps for each of these relaxed nodes is at least $S(t)=\gamma' 2^{t} 2^{n-t-t_r} (n-t-t_r)$ polarization steps. The bound follows by summation of $S(t)$ for all $t$ with $t_b \leq t \leq n - t_r $. Notice that since the right-polarized children of those odd-indexed nodes at $t$ are even indexed, they are not counted among the odd-indexed nodes in $\mathcal{G}_{B,t'}$, for any other $t'$.
\end{proof}

%\clearpage
 \begin{figure}[t!]
\centering
\includegraphics[width=0.8\linewidth]{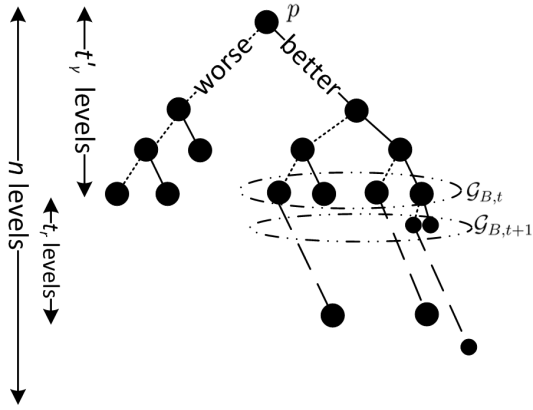}
\caption{Lower bound (LB2) on Complexity Reduction.}
\label{fig:LB2}
\end{figure}

Notice that in the above lower bound, a necessary condition is that $t + 2 t_r > n$ to guarantee that no double counting occurs. In many practical operation scenarios, this condition holds. If the condition does not hold, one can modify the bound by limiting the computed summation to $\sum_{t_0=t}^{t + t_r}  S(t_0)$.

\begin{remark}
In \Tref{Th:GCLB1}, for large enough $t_{\gamma}$, the parameter $\gamma$ is independent of $B$ and is approximately equal to $C$, the capacity of the underlying channel $W$. Also, $\gamma'$ in \Tref{Th:GCLB2}, is approximately $C/2$. For our practical applications, we can always pick a proper value of $B$ such that these approximations still hold at the desired block length.
\end{remark}

\subsection{Analysis of complexity reduction for AC-RP}

In this subsection, we analyze the complexity reduction from bad-channel relaxed polarization, as well as all-channel (both good and bad) relaxed polarization. As opposed to the previous subsection, we limit our attention in this subsection to binary erasure channels (BEC), wherein the exact computation of Bhattacharrya parameters is applicable at finite block lengths.

Throughout this subsection, we always assume the channel BEC$(p)$. For a function $F(p)$, let $F^t(p)$ denote the output from recursive application of the function $F$ $t$ times, with initial input $p$.

\begin{theorem}
\label{Th:BCUB}
The bad-channel relaxed polarization complexity reduction is upper bounded by
\eq{\mathcal{R}_b(n) \leq (n-t_b)/n,}
where $t_b = \min_t \{ F^t(p) > 1 - \mathcal{T} \}$ and $F(p)=2p-p^2$. \end{theorem}

\begin{proof} The left child of a node with BP $Z$ is associated with a bit-channel with BP $F(Z)$. Hence, it requires $t_b$ polarization levels for the worst left-polarized bit-channel to have BP greater than  $1 - \mathcal{T}$. The rest of the proof follows as for the GC-RP case.
\end{proof}

Theorem \ref{Th:BCUB} can also be proved by combining the results of Theorem \ref{Th:duality} and Theorem \ref{Th:GCRPUB}. The bounds derived for the good-channel relaxed polarization in the previous subsection, can be turned into bounds for bad-channel relaxed polarization of BEC($p$) by replacing $p$ with $1-p$ in the bounds, and modifying other parameters accordingly. Hence, to avoid writing similar proofs, we only mention the theorems and skip the proofs.

Let $t_g = \left \lceil \log_2 \frac{\log_2 \mathcal{T}}{\log_2 p}\right \rceil$ and $t_b$ as in \Tref{Th:BCUB}. In fact, $t_b= \left \lceil \log_2 \frac{\log_2 \mathcal{T}}{\log_2{1-p}} \right \rceil$. Observe that if $p > 0.5$, then $t_b < t_g$, if $p < 0.5$, then $t_b > t_g$, and if $p = 0.5$, then $t_b = t_g$. Combining \Tref{Th:BCUB} with upper-bounds on GC-RP of Theorem \ref{Th:GCRPUB} results in the following upper bound on AC-RP complexity reduction.

\begin{corollary} The all-channel relaxed polarization complexity reduction is upper bounded by
\eq{\mathcal{R}_a(n) \leq (n-t')/n,}
where $t' = t_g$ for $p \leq 0.5$ and $t' = t_b$ for $p > 0.5$. \end{corollary}

The next theorem can be also regarded as the counterpart of \Tref{Th:GCLB1}, for bad-channel relaxed polarization.

\begin{theorem} \label{Th:BCLB1} Let $t_{c} = \left \lceil \log_2  \frac{\log_2 G}{\log_2 (1 - p)} \right \rceil$, for an arbitrary threshold $G$. For a polarization level $t_{\beta}  \geq t_c$, let $\beta = |\{i: Z_{i,t_{\beta} } \geq 1 - G \}|/2^{t_{\beta}} $. Then, the bad-channel relaxed polarization complexity reduction is lower bounded by
\eqn{\mathcal{R}_g(n)  \geq  \beta 2^{-t_l}(n-t_l-t_{\beta} )/n}
where $t_l= \left \lceil \log_2 \frac{\log_2 \mathcal{T}}{\log_2{G}} \right \rceil$. \end{theorem}

For AC-RP, the lower bounds of Theorem \ref{Th:GCLB1} and Theorem \ref{Th:BCLB1} can be combined as in the following corollary. It is assumed that the set of relaxed nodes in  GC-RP and the set of relaxed nodes in BC-RP do not intersect. This is a valid assumption as long as the good-channel and bad-channel relaxation thresholds, $\cT$ and $1-\cT$, are far apart enough, as characterized in subsection \ref{subsec: BP}.

\begin{corollary} \label{Th:ACLB1} The all-channel relaxed polarization is lower bounded by
\eqn{ \mathcal{R}_a(n)  \geq \frac{1}{n} \left( \beta 2^{-t_l}(n-t_l-t_{\beta})  + \gamma 2^{-t_r}(n-t_r-t_{\gamma}) \right).}
 \end{corollary}

\begin{theorem}\label{Th:BCLB2}  Let $t_{c} = \left \lceil \log_2  \frac{\log_2 G - 1}{\log_2 (1 - p)} \right \rceil + 1$, for an arbitrary threshold $G$. For a polarization level $t \geq t_c$, let $\beta' = \min_{t \geq t_c} |\left\{i\ \text{odd}:  Z_{i,t} \geq 1 - G \right\}|/2^t$.  Then, the bad-channel relaxed polarization complexity reduction is lower bounded by
\eqn{\mathcal{R}_b(n)  \geq  \frac{\beta' 2^{-t_l}}{2n} \left( (n-t_l)(n-t_l-2t + 1) + t(t-1) \right)}
where $t_l = \left \lceil \log_2 \frac{\log_2 \mathcal{T}}{\log_2 G} \right \rceil$.
\end{theorem}

Similar to \Tref{Th:GCLB2}, the above theorem holds under the condition that $t+ 2 t_l > n$. Also, similar to Corollary \ref{Th:ACLB2}, the following corollary holds by combining Theorem \ref{Th:GCLB2} and Theorem \ref{Th:BCLB2}. To make notations consistent, let $t'_{\gamma}$ denote the level $t$ in \Tref{Th:GCLB2} and $t'_{\beta}$ denote the level $t$ in \Tref{Th:BCLB2}.

\begin{corollary} \label{Th:ACLB2} The all-channel relaxed polarization is lower bounded by
\eqn{\begin{split} \mathcal{R}_a(n)   \geq  & \frac{\beta' 2^{-t_l}}{2n} \left( (n-t_l)(n-t_l-2t'_{\beta} + 1) + {t'_{\beta}}^2-t'_{\beta} \right) \\ & +
\frac{\gamma' 2^{-t_r}}{2n} \left( (n-t_r)(n-t_r-2t'_{\gamma} + 1) + {t'_{\gamma}}^2-t'_{\gamma} \right).\end{split}
}
 \end{corollary}

%=======================================================================%
%                                                                       %
%     5. Numerical Evaluation of Relaxed Polarization on Erasure Channels                       %
%                                                                       %
%=======================================================================%

\subsection{Numerical evaluation of  complexity reduction by relaxed polarization on erasure channels}
\label{sec:sec5}

In this subsection, we compute the complexity reduction of different scenarios of relaxed polarization over binary erasure channels and compare them with the bounds provided in this section.

The block length of the constructed relaxed polar code is assumed to be $N = 2^{20}$ and the FER of $E = 10^{-5}$ is assumed for the code construction. The erasure probability $p$ will be varying between $0$ and $1$. We have observed that the thresholds $B = 2p - p^2$ and $1 - G = p^2$ will result in desired values for $\gamma$ in \Tref{Th:GCLB1} and $\beta$ in \Tref{Th:BCLB1}. With these values of $B$ and $G$ we have $\gamma = 1 - p$ and $\beta = p$.  We have also observed that $\gamma'$ in \Tref{Th:GCLB2} and $\beta'$ in \Tref{Th:BCLB2} can be well approximated by $\gamma/2$ and $\beta/2$. The results of  \figref{fig:g1} show that actual CR of GC-RP can be characterized using the derived upper and lower bounds, and up to $70\%$ complexity reduction is achievable at a target FER of $10^{-5}$.  

  The performance of AC-RP is analyzed in \figref{fig:A1} at the same target FER of $10^{-5}$, where the analytical bounds are compared to the numerical results from actual construction. It is observed that the bounds give a good approximation of the actual complexity reduction.  GC-RP is effective with good channel parameters and BC-RP is more effective with bad channel parameters. The bounds are also minimized at $p=0.5$, and symmetric around $p=0.5$. This can be explained by the symmetry property of Theorem \ref{Th:duality}.
	
 \begin{figure}[t!]
\centering
\includegraphics[width=\linewidth]{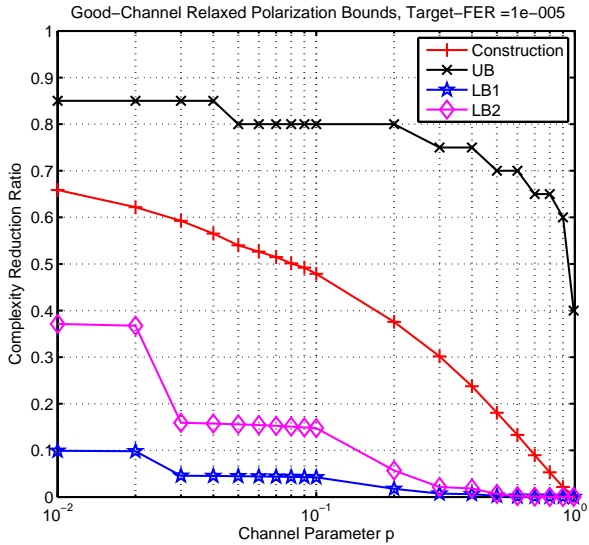} 
\caption{Bounds on the good-channel relaxed polarization complexity reduction with target FER of $10^{-5}$ .}
\label{fig:g1}
\end{figure}

%=======================================================================%
%                                                                       %
%     6. Relaxed Polarization on General BMS Channels                                                                    %
%                                                                       %
%=======================================================================%

\section{Relaxed Polarization on General BMS Channels}
\label{sec:sec6}

In this section, we describe how a code can be constructed and decoded on general binary memoryless channels using relaxed polarization.

\subsection{Construction of relaxed polar codes on general BMS channels}
For general binary memoryless channels, the Bhattacharyaa parameters are exponentially hard to compute as block length increases. This is due to the exponential output alphabet size of the polarized bit-channels. Instead of exact calculation of Bhattacharyya parameters, they can be well approximated by bounding the output alphabet size of bit-channels via channel degrading and channel upgrading transformations \cite{TalVardy13}.  The channel degrading and upgrading operations provide tight lower bounds and upper bounds on the corresponding Bhattacharyya parameters. In order to construct polar codes for continuous-output  BMS channels (e.g.  additive white Gaussian noise (AWGN) channels), the channel can be first quantized. Then, the degrading and upgrading operations will be performed for the bit-channels resulting from polarization of the quantized channel \cite{TalVardy13}.
For AWGN channels, the bit channel error probability (BC-EP) can also be reasonably approximated using density evolution and
a Gaussian approximation \cite{Trifonov}. Alternatively, for short codes, the BC-EP can be numerically evaluated using
Monte-Carlo simulations, assuming a genie-aided SC decoder. For generality of description, let the error probability of the $j$-th bit-channel at the $t$-th polarization level be bounded by
\eq{ \underline{E}_{t,j} \leq E_{t,j} \leq \overline{E}_{t,j}}
where $\underline{E}_{t,j}$ is the probability of error of the upgraded version of $W^{(j)}_{2^t}$ and $\overline{E}_{t,j}$ is the probability of error of its degraded version.

The values of $\underline{E}_{t,j}$ and $\overline{E}_{t,j}$ can be computed using upgraded and degraded versions of polarization tree. In the upgraded polarization tree, after each level of polarization the resulting bit-channels will be upgraded to have a limited output alphabet size. At the next level, the upgraded bit-channels will be polarized. As a result, all the bit-channels in the upgraded polarization tree will have a limited output alphabet size. Therefore,  $\underline{E}_{t,j}$ can be easily computed. The same procedure is repeated to get a degraded polarization tree and to compute $\overline{E}_{t,j}$. The construction of fully polarized codes can be done according to either lower bounds or upper bounds on the probability of error of the bit-channels at the last level of polarization. For instance, in case of using upper bounds, bit-channel are sorted according to their error probabilities $\overline{E}_{n,j}$ in ascending order. Accumulate as many good bit-channels in the set $\Gamma$, such that $\sum_{j\in \Gamma}\overline{E}_{n,j} \leq E$, where $E$ is the target FER. Then, the FP code is defined by $\Gamma$ and has rate $R=|\Gamma|/N$.

In proposed good-channel relaxed polar codes, a node will not be further polarized if the upper bound on its bit-channel error probability is lower than a certain threshold $\mathcal{E}_g$. For bad-channel relaxed polar codes, a bit-channel will not be further polarized if the lower bound on its error probability exceeds an upper threshold $\mathcal{E}_b$. Numerically, it was found for BMS channels that the error performance of the constructed code is closer to that of the upper-bound calculated using the degraded channel. Hence, when polarization is relaxed for a node, the error probability of the children of a non-polarized node is set to the degraded-channel error probability of the parent.  As a result, the procedure for designing relaxed polar codes of length $N$ at a target FER $E$ on general BMS channels is specified below. Each node $j$ at level $t$ in the polarization tree is associated with a label Relaxed($t,j$) which is initialized to 0, and will be set to 1 only if this node will not be polarized. The error probability (EP) of each node in the RP tree is initialized to that of the fully polarized tree $\overline{E}^R_{t,j} = \overline{E}_{t,j}$. The relaxed polar code will be defined by its good channel set $\Gamma_{R}$.

\begin{spacing}{1}
\begin{algorithm}[H]
\caption{Relaxed Construction for General BMS Channels}
\begin{algorithmic}[1]
    \STATE \textbf{Stage 1:} Calculate AC-RP bit-channel EP for target FER $E$ and rate $R$
    \bindent
    \STATE Set the GC relaxation threshold as $\mathcal{E}_g = E/(R l)$
    \STATE Set the BC relaxation threshold as $\mathcal{E}_b = H^{-1}\left( 1 - H(\mathcal{E}_g) \right)$
    		    \FOR{$t=1:n$}
    \FOR {$j=1:2^t$}
		    \IF{Relaxed$(t-1,\lceil j/2 \rceil)=1$}
    \STATE Relaxed$(t,j) =1$, $\overline{E}^R_{t,j} = \overline{E}^R_{t,\lceil j/2 \rceil}$
    \ELSIF{$\left\{ \overline{E}_{t,j} < \mathcal{E}_g \right\}$ \OR $\left\{ \underline{E}_{t,j} > \mathcal{E}_b \right\}$}
    \STATE Relaxed$(t,j) =1$
    \ENDIF
				\ENDFOR
    \ENDFOR
         \eindent
      \STATE \textbf{Stage 2:} Construct AC-RP code with rate $R$
        \bindent
    \STATE Sort bit-channel EPs $\overline{E}^R_{t,j}$ in ascending order
    \STATE Select $\Gamma_{R}$ to have the $R N$ bit channels with the smallest EP
      \eindent
\end{algorithmic}
\end{algorithm}
\end{spacing}

  \begin{figure}[t!]
\centering
\includegraphics[width=\linewidth]{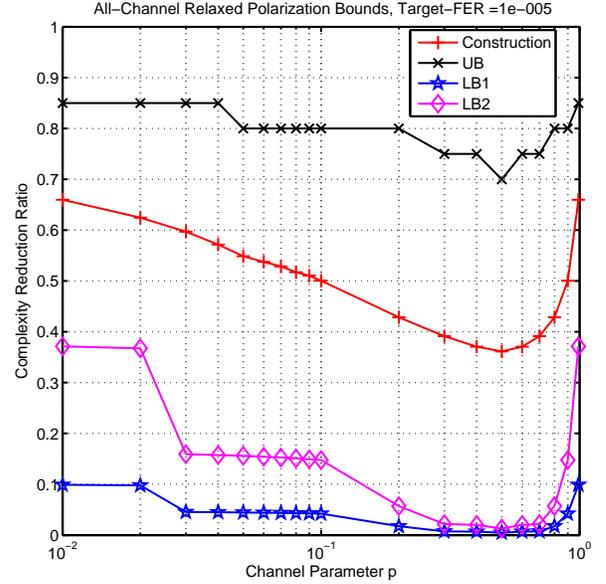} 
\caption{Bounds on the all-channel relaxed polarization complexity reduction with target FER of $10^{-5}$ .}
\label{fig:A1}
%\vspace{-0.5cm}
\end{figure}

For the case of AWGN channel, first the channel parameter is calculated to satisfy the condition $C \geq R$, where $R$ is the target rate and $C$ is the capacity of the channel. Then, the channel is quantized using the method of \cite{TalVardy13} to get a channel with discrete output alphabet. Then Algorithm 1, discussed above, will be applied to this channel.

With target FER $E$, the good-channel relaxation threshold is chosen to be $\mathcal{E}_g = E/(R N)$ to satisfy the target FER, i.e. $\sum_{j \in \Gamma_R} \overline{E}_{n,j} \leq E$. Let $H(E_W)$ be the entropy of the channel $W$ with fidelty $E_W$, such that $H(E_W)=1- I(W)$, where $I(W)$ is capacity of channel $W$.
Then, the bad-channel relaxation threshold is chosen such that $H(\mathcal{E}_b) =  1 - H(\mathcal{E}_g) $.
For general BMS channels $W$ with error probability $E_W$,  the approximation $H(W) \approx h_2(E_W)$ can be used, based on the inequality $H(W) \leq h_2(E_W)$ \cite{arikan2009channel, richardson2008modern}, where $h_2(p)=-p \log_2(p) - (1-p) \log_2(1 -  p)$ is the binary entropy function.
To guarantee that there is no rate loss from bad-channel relaxation if $\left\{ \underline{E}_{t,j} > \mathcal{E}_b \right\}$ is satisfied at Step 13, then relaxation may only be done after verifying that the lower bound on the error probability of the best upgraded descendent channel of that node is still higher than $\mathcal{E}_g$, which will be satisfied for practical frame errors $E$.

The same procedure above can be used to construct GC-RP codes, by
neglecting the bad-channel relaxation condition $\left\{ \underline{E}_{t,j} > \mathcal{E}_b \right\}$ in step 13.

In case of erasure channels with erasure probability $p$, the channel parameter (erasure probability) for the target channel capacity is  $p = 1 - C$. The upper and lower bound on the bit-channel error probabilities coincide, and can be calculated exactly by the BPs, $E_{t,j} = {Z}_{t,j}/2$. To calculate the relaxation thresholds, the error probability is $E=p/2$, and the entropy is $H(E)=2E$.

Algorithm 1 constructs the relaxed polar codes for general BMS channels, using bounds on the bit-channel error probability. However, for short block lengths, the bit-channel error probability can be numerically calculated to be $\widetilde{E}_{t,j}$  using Monte-Carlo simulations, assuming a genie-aided successive cancellation decoder. In such a case, Algorithm 1 is modified by letting $ \underline{E}_{t,j}  = \overline{E}_{t,j} = \widetilde{E}_{t,j}.$

 \subsection{Decoding of relaxed polar codes on general BMS channels}
 \label{sec:dec}

Decoding of relaxed polar codes can be done by a modified successive cancellation decoder. For a polar code of length $N$ and BMS channel $W$, suppose that $u^N_1$ is the input vector and $y^N_1$ is the received word.

 Consider a relaxed polar code constructed as explained in the previous subsection. At each level $t=\log N'$, for  $1 \leq i \leq N$, Relaxed$(t,i)=1$ means that $\widetilde{W}_{N'}^{(i)}$ is not polarized and Relaxed$(t,i)=0$ means that $\widetilde{W}_{N'}^{(i)}$ is fully polarized. In practical application of relaxed polar codes, the decoder will have prior knowledge of the polarization map by Relaxed$(t,i)$, which requires at most storage of $2N$ bits.  For communication systems, the polarization map can be specified by the communication standard,  similar to the specification of the parity-check matrices of block codes.

 For $i=1,2,\ldots, N$, the decoder computes the likelihood ratio (LR) $L^{(i)}_N$ of $u_i$, given the channel outputs $y^N_1$ and previously decoded bits $\hat{u}^{i-1}_1$
$$
L^{(i)}_N(y^N_1,\hat{u}^{i-1}_1) = \frac{\widetilde{W}^{(i)}_N (y^{N}_1,\hat{u}^{i-1}_1 | u_i = 0) }{\widetilde{W}^{(i)}_N (y^{N}_1,\hat{u}^{i-1}_1 | u_i = 1)}.$$

For FP polar codes, Ar{\i}kan observed that calculation of the LRs at length $N$ require another $N$ calculations at the parent node at length $(N/2)$, where the LRs from the pair
$\left(L^{(2i-1)}_N (y^N_1,\hatu^{2i-2}_1), L^{(2i)}_N(y^N_1,\hatu^{2i-1}_1) \right)$ are assembled from
the pair $\left( L^{(i)}_{N/2}(y^{N/2}_1,\hatu^{2i-2}_{1,e} \oplus \hatu^{2i-2}_{1,o}), L^{(i)}_{N/2}(y^N_{N/2+1},\hatu^{2i-2}_{1,e}) \right)$, via a straightforward calculation using the bit-channel recursion formulas for $n \geq 1$ \cite{arikan2009channel}.
The relaxed successive cancellation decoder (RSCD) follows the same recursion.
Hence, if Relaxed$(t,i)=0$, the likelihood ratio (LR) $L^{(i)}_N$ can be computed recursively as follows.
  \begin{align}
\label{LLR_odd}
&L^{(2i-1)}_N (y^N_1,\hatu^{2i-2}_1)\\
\notag
&= \frac{1+L^{(i)}_{N/2}(y^{N/2}_1,\hatu^{2i-2}_{1,e} \oplus \hatu^{2i-2}_{1,o})
L^{(i)}_{N/2}(y^{N}_{N/2+1},\hatu^{2i-2}_{1,e})}{L^{(i)}_{N/2}(y^{N/2}_1,\hatu^{2i-2}_{1,e} \oplus \hatu^{2i-2}_{1,o})
+L^{(i)}_{N/2}(y^{N}_{N/2+1},\hatu^{2i-2}_{1,e})},\\
\label{LLR_even}
&L^{(2i)}_N(y^N_1,\hatu^{2i-1}_1)\\
\notag
 &= \left[ L^{(i)}_{N/2}(y^{N/2}_1,\hatu^{2i-2}_{1,e} \oplus \hatu^{2i-2}_{1,o}) \right]^{1-2\hatu_{2i-1}}
L^{(i)}_{N/2}(y^N_{N/2+1},\hatu^{2i-2}_{1,e}).
\end{align}
Otherwise, if Relaxed$(t,i)=1$ the decoding equations are modified as follows:
 \begin{align}
\label{LLR_odd_new}
L^{(2i-1)}_N (y^N_1,\hatu^{2i-2}_1) &= L^{(i)}_{N/2}(y^{N/2}_1,\hatu^{2i-2}_{1,o})\\
\label{LLR_even_new}
L^{(2i)}_N(y^N_1,\hatu^{2i-1}_1) &= L^{(i)}_{N/2}(y^N_{N/2+1},\hatu^{2i-2}_{1,e}).
\end{align}
The hard-decision estimates at a parent node are calculated from the hard-decision estimates of its two children in a step similar to encoding.
At the last stage when $N=1$, the LRs are simply $L_1^{(1)}(y_i) = W(y_i|0)/W(y_i|1)$. At the end, hard decisions are made on $L_N^{(i)}$ (at the leaf nodes), except for frozen bit-channels $W^{(i)}_N$ where $\hat{u_i} = u_i = 0$.

 From the above description, $2N$ LR calculations for $\log N$ levels are required for decoding conventional FP polar codes. However, the decoding complexity of relaxed polar codes is linearly reduced by the ratio of relaxed nodes in the polarization tree, since no LR calculation is required at relaxed nodes.

 The relaxed successive cancellation decoding as discussed above is extended to perform the relaxed successive cancellation list (SCL) decoding of RP codes. The SCL decoding of polar codes is shown to have considerable improvement over the regular SC decoding \cite{TVlist, NiuCRC, NiuSurvey}. In SCL decoder, instead of only one path, i.e. a sequence of decoded information bits, up to $L$ decoding paths are considered at each decoding stage. The decoding paths are being updated as the decoder evolves. At each stage of the decoding process, where an information bit $u_i$ is being decoded, both options of $u_i = 0$ and $u_i=1$ are being considered and hence, the number of decoding paths is doubled to at most $2L$. Then this extended list of size up to $2L$ is pruned, based on a maximum likelihood metric, to get a list of size $L$ of the locally most likely paths. In the SCL decoding, there are up to $L$ likelihood ratios of $L_N^{(i)}$ at each node in the decoding trellis and then up to $L$ parallel recursive calculations, as in \eqref{LLR_odd} and \eqref{LLR_even}, are performed at the node. In the relaxed SCL decoding, if a node is relaxed, then there will be $L$ parallel decoding equations as in \eqref{LLR_odd_new} and \eqref{LLR_even_new}. The operations for picking the most likely paths remain the same for relaxed SCL.

\subsection{Relaxed polarization versus simplified successive cancellation decoding}
\label{sec:sec5b}
Whereas relaxed polarization results in a construction of a code different from that obtained by Ar{\i}kan's full polarization, the SSCD \cite{alamdar2011simplified} is a simplified decoder for a specific code. In fact, as would be clarified below, the SSCD can also be used to further reduce the complexity and latency of decoding relaxed polar codes.

By construction, relaxed polarization attempts to maximize the number of rate-one nodes and rate-zero nodes by relaxing the polarization of sufficiently good and sufficiently bad bit-channels, respectively. Rate-1 and rate-0 nodes are nodes which have all their descendants in the good channel set and  the bad-channel set, respectively.
As was clarified in Section \ref{sec:dec}, no encoding or decoding operations are done at the relaxed nodes.

The SSCD identifies the rate-1 and rate-0 nodes in the received code, and reduces the operations required to decode the corresponding constituent codes. Hence, SSCD does not offer complexity reductions at the encoder. Since a rate-0  node only has frozen bits at its output, its constituent tree does not need to be traversed when decoding since the leaf values are known a priori. The output bits of the tree rooted at a rate-1 node can be found by simple hard-decisions on the soft likelihood ratios at the rate-1 node. However, since these bit-channels were polarized at the encoder, the input bits at the rate-1 node need to be recovered with a step similar to re-encoding in order to recover the information bits.

 Hence, relaxed polarization offers complexity and latency reductions at both the encoder and the decoder, while SSCD only reduces the decoding complexity and latency compared to Ar{\i}kan's successive cancellation decoder. The decoding complexity reduction is calculated for the SSCD as for the relaxed code, where rate-1 nodes and rate-0 nodes contribute to the complexity reduction same as relaxed nodes, and the re-encoding complexity at the rate-1 nodes is neglected.

 Next, we compare the  reductions in decoding latency. As described in the previous section, a polarized node requires three clock cycles to calculate the even and odd LRs, and then calculate its hard-decision estimate from the hard-decision estimates of its children pair using the encoding operation. Hence, the total decoding latency with successive cancellation decoding for a polar code of length $N=2^n$ can be assumed to be $\mathcal{L}(n) = 3\sum_{i=0}^{n-1} 2^i = 3N -3$. Consider the RP code decoded with the RSCD. A BC relaxed node requires no operation and hence contributes nothing to the latency. A sub-tree of GC relaxed nodes requires only one clock cycle at  its root to calculate its hard-decision estimates. Similarly, for the SSCD, rate-0 nodes contribute nothing to the latency, and a sub-tree of rate-1 nodes requires one clock cycle. However, since this rate-1 constituent code is fully polarized, if the root of the rate-1 constituent code is at level $n-t$, then additional $t$ clock cycles are required for re-encoding to recover the information bits at the leaf nodes.

 To this end, there are two important observations to make.

 Firstly, the SSCD can be combined with RSCD to decode RP codes. The SSCD as proposed in  \cite{alamdar2011simplified} is applied on Ar{\i}kan's fully polarized code, and will be noted as SSCD FP. Since the relaxed polar construction as described above relaxes the nodes before determining the good-channel set, then there can exist rate-1 and rate-0 nodes in the resultant code which have not been relaxed. Then, this implies that SSCD can also be used to decode relaxed polar codes, where after determining the good-channel set, the rate-1 and rate-0 nodes are identified and the operations at rate-1 nodes and rate-0 nodes which have not be relaxed by construction will be simplified as in SSCD. Hence, combined SSCD and RSCD on RP codes, denoted by SSCD RP,  will further reduce the decoding complexity of RP codes. Moreover, since RP codes are constructed to have more rate-1 and rate-0 nodes by the relaxation operation,  SSCD RP will often have reduced decoding complexity compared to SSCD FP.

 Secondly, the relaxed polar code construction can be modified such that all rate-1 and rate-0 nodes of the fully polarized code are relaxed at the encoder. This modified relaxed polarization (MRP) construction is done by first constructing Ar{\i}kan's fully polarized code, selecting the good-channel set according to the desired rate or target error probability, finding the modified relaxed polarization map as that which relaxes all the rate-1 and rate-0 nodes in the FP code, and then encoding according to the modified relaxed polarization map. The good channel set for the MRP code will be fixed as that in the FP code. If only the rate-1 nodes are relaxed, then the code is called GC-MRP. If both the rate-1 and rate-0 nodes are relaxed, the code is called AC-MRP. Since all rate-1 and rate-0 nodes are already relaxed, the combined SSCD-RSCD will not produce additional complexity or latency reductions compared to RSCD when decoding the MRP code.  Furthermore, neglecting the re-encoding complexity required by the SSCD at rate-1 nodes, RSCD on the modified RP codes will have the same decoding complexity but lower decoding latency compared to that of SSCD on FP codes. MRP codes also have the additional advantage of having lower error rates than their corresponding FP codes, as shown in Lemma \ref{lem:error}.

The complexity reductions by RSCD RP, SSCD FP, SSCD RP are compared in  \figref{fig:CRSSCD} for two cases: GC-RP versus SSCD when applied to rate-1 nodes only, denoted by SSCD(1), and AC-RP versus SSCD when applied to both rate-1 and rate-0 nodes, denoted by SSCD. The complexity reductions are shown for a code of length $2^{16}$ on the binary erasure channel. The FP codes are constructed at a rate equals to $9/10$ of the channel capacity. To construct the RP codes, the error probability $E$ of the good-channel set of the FP code is calculated and used to calculate the relaxation thresholds by $\mathcal{T}_g = 2E/| \Gamma |$, and $\mathcal{T}_b = 1-\mathcal{T}_g$. The asymptotic bound of CR with GC-RP is the capacity as proved in Theorem \ref{asymp-bound1}, and is also shown. It can be noted that the results of the figure are inline with the discussion above. When considering the GC-RP and rate-1 nodes only, RSCD GC-RP can offer higher complexity reduction than SSCD(1), especially at higher rates. However, after taking the bad-channel nodes and the rate-0 nodes into account, SSCD FP has higher complexity reduction than RSCD RP  which implies that the bad channel relaxation threshold can be made more aggressive without degrading the performance. Across all rates, the combined RSCD-SSCD on RP codes has the highest complexity reduction.
RSCD on the MRP codes is shown in \figref{fig:LRSSCD} to have the least latency compared to SSCD on FP codes, and  RSCD on RP codes. It is observed that in case of GC-RP, the decoding latency decreases at higher code rates due to the increase in the number of relaxed nodes. For AC-RP, the latency increases with the code rate since GC-relaxed (or rate-1) nodes require more latency than BC-relaxed (or rate-0) nodes, as described above. It has also been observed that the percentage of latency reduction increases with the code length.

\begin{figure}[th!]
\centering
\includegraphics[width=\linewidth]{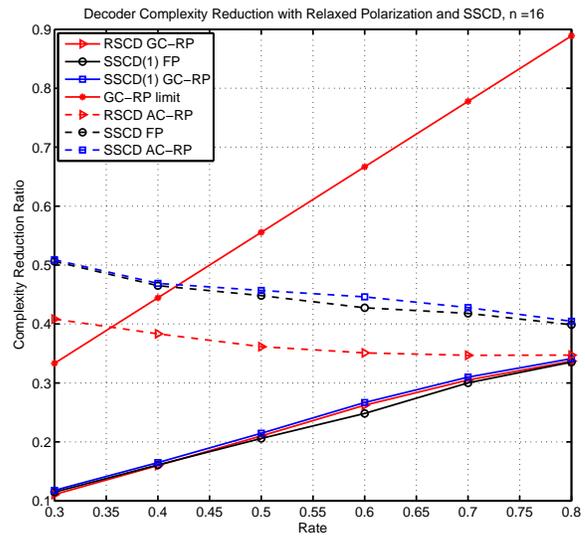} 
\caption{Complexity reduction ratios of RP codes and SSCD at different code rates and code length $2^{16}$ on the BEC channel.}
\label{fig:CRSSCD}
%\vspace{-0.5cm}
\end{figure}

\begin{figure}[th!]
\centering
\includegraphics[width=\linewidth]{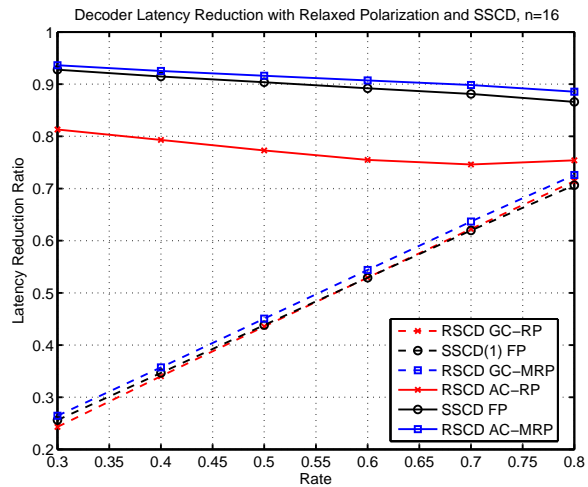} 
\caption{Latency reduction ratios of RP codes and SSCD at different code rates and code length $2^{16}$ on the BEC channel.}
\label{fig:LRSSCD}
%\vspace{-0.5cm}
\end{figure}

 \subsection{Performance on the  AWGN channel}\label{sec:AWGN}

 The achievable complexity reduction is analyzed by actual construction of the relaxed polar codes on AWGN channels in \figref{fig:AWGNCR}. An AWGN channel with binary-input capacity $C=0.7$ is used to calculate upper and lower bounds on the bit-channel error probabilities.
 The CR at different code lengths $2^{10}$, and $2^{14}$ are logged at different target FER $E$. The rate, achievable by construction of the FP code at each target FER $E$, is also logged. It is observed that at a larger target FER $E$, a higher rate is possible, due to possible accumulation of more good-channel bits. The CR also increases with the target FER $E$, due to the increase of the relaxation threshold $\mathcal{E}_g$, despite the increase in the code rate. Since the number of polarization levels increases with the code block length $N$, the achievable CR from RP increases with $N$. The effect of CR due to bad-channel relaxation becomes more visible, at higher target FER, as $\mathcal{E}_b$ also becomes lower. 

\begin{figure}[t!]
\centering
\includegraphics[height=3in, width=\linewidth]{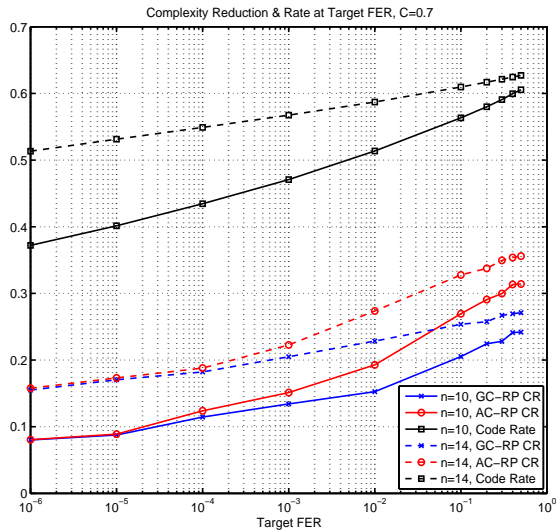} 
\caption{Complexity reduction ratios on AWGN channel at code lengths $N=10$, and $N=14$ at different code rates}
\label{fig:AWGNCR}
%\vspace{-0.5cm}
\end{figure}

The error-rate performance by actual relaxed successive cancellation decoding of the RP codes is shown in \figref{fig:AWGN1}, for binary phase shift keying (BPSK) on an AWGN channel with variance $\sigma^2$ as a function of the signal to noise ratio SNR$=10 \log_{10} (1/\sigma^2)$. Practical code length of $N=4096$ is assumed with near half-rate code of $R=0.59$, respectively. The code is constructed with Algorithm 1, assuming an AWGN channel with binary-input capacity $C=0.7$, and with $E = 0.1$.
For simpler comparison, both the GC-RP and AC-RP codes use the same good bit-channel set found by construction (Stage 4) of the AC-RP code. However, the FP good-channel set is optimized for the FP polar code.  It is observed that the frame error rate (FER) and bit error rate (BER)  performances of the GC-RP code are similar to those of the AC-RP code. Although, the RP codes can have a slightly higher FER than the FP code due to different information sets, it is observed that the RP codes have lower information bit error rates than the FP codes. This verifies the proper design and selection of relaxation thresholds for the RP codes. Another important observation is that the proposed construction of relaxed polar codes is robust enough, such that it performs well over the whole range of simulated SNRs, although the codes are constructed for a certain SNR.

\begin{figure}[t!]
\centering
\includegraphics[height=3in, width=\linewidth]{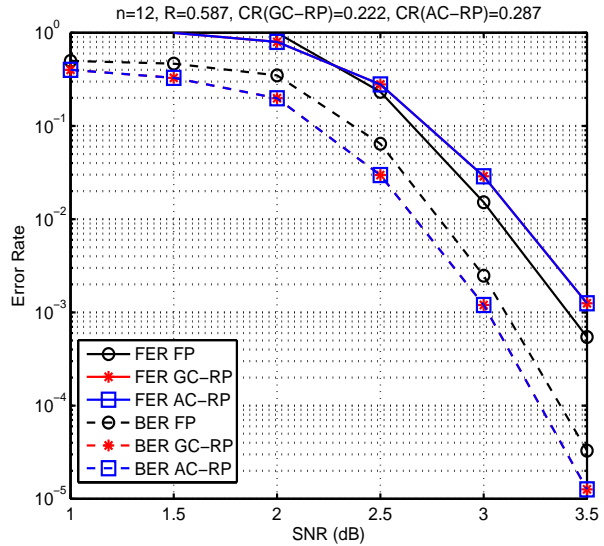}%{Res_n12_R0554_ep01b.eps} %{BLCC.eps}\\%{LM2SMLM.eps}\\
\caption{Error performance of relaxed polar codes by BPSK on AWGN channel. The code length is $2^{12}$ and the code rate is $0.59$
}
\label{fig:AWGN1}
\end{figure}

The performance of the relaxed SCL decoder is compared with the performance of the regular SCL decoder for a numerical example. The simulations are done for the code block length $N = 1024$, and rate $R = 0.5$. First, the fully polarized polar code of rate $0.5$ is constructed for an AWGN channel at an $\mbox{SNR}=2$ dB. Then, the all-channel modified relaxed polar code is constructed by considering the same set of information bit indices. The complexity reduction ratio of the modified relaxed polar code from the good channel relaxation only is $0.1639$ and from the all channel relaxation is $0.3340$. Regular SC decoding of the FP code and RSCD of the RP code are simulated and compared over the AWGN channel for the constructed codes, which corresponds to the case with list size $L=1$. Furthermore, the relaxed SCL decoder, as discussed in Section \ref{sec:dec}, and the regular SCL decoder, with a maximum list size of 32 are simulated and compared as well. For list decoding, the polar information bits are concatenated with a CRC code of length 16, where the rates are adjusted accordingly so that the actual information rate is $0.5$, i.e., the information block length of the polar-CRC code is increased to $512+16=528$. The simulation results are shown in \figref{fig:RSCLD} and show about $1$ dB SNR gain with list decoding. It is observed that with successive cancellation decoding, the RP code has a slightly better FER than the FP code. Since the RP code has the same information set as the FP code, this can be justified by Lemma \ref{lem:error}. Furthermore, a better bit error rate (BER), with up to $0.2$ dB SNR, is observed for the RP code compared to FP code.

\begin{figure}[t!]
\centering
\includegraphics[width=\linewidth]{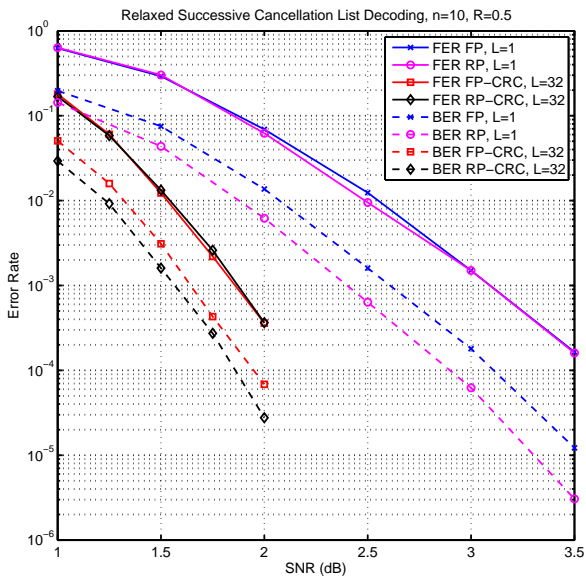} %{BLCC.eps}\\%{LM2SMLM.eps}\\
\caption{Performance of relaxed successive cancellation list decoding on AWGN Channels, the code length is $2^{10}$ and the polar code rate is $0.5$ at $L=1$ without CRC bits.}
\label{fig:RSCLD}
%\vspace{-0.5cm}
\end{figure}

  \section{Conclusion}
  \label{sec:sec7}
  In this work, a new paradigm for polar codes, called relaxed polar coding, is investigated. In relaxed polar codes, a bit-channel will not be further polarized if it has been already polarized to be sufficiently good or sufficiently bad. Hence, encoding and decoding of relaxed polar (RP) codes have lower computational and time complexities than those of conventional polar codes. RP codes also have lower space complexity than conventional polar codes in fixed point hardware implementations, due to the less number of bits required to store the likelihood ratios. This has the compound effect of decoder implementations with less power consumption. It is proved in this work that, similar to conventional polar codes, RP codes are capacity achieving. It is also shown that with proper design, RP codes will have lower  error rates than conventional polar codes of the same rate. Constructions of RP codes on the binary erasure channel, and on general BMS channels are described. Asymptotic and finite-length bounds on the complexity reduction achievable by relaxed polar coding are derived and verified for the binary erasure channel against actual constructions. The relaxed successive cancellation decoding (RSCD) for relaxed polar codes is described. The successive cancellation list decoder for polar-CRC codes \cite{TVlist, NiuCRC} is also modified for list-decoding of relaxed polar-CRC codes. Moreover, we discuss how simplified successive cancellation decoding \cite{alamdar2011simplified}, can be done on top of RSCD to further reduce the decoding complexity and latency of RP codes.
For a code of rate $0.3$ and length $2^{16}$, the results show an $50.9\%$ decoding complexity reduction ratio and an $93.5\%$  decoding latency reduction ratio are possible with relaxed successive cancellation decoding of RP codes on the BEC. 	
 It is verified by numerical simulations on the AWGN channel that the information bit error rates of properly designed RP codes are at least as good as those of conventional polar code with the same rate.

Next, we discuss possible directions for future work. 	
Whereas the derived bounds on the complexity reduction ratios on the BEC channel have explicit closed form formulas, the numerical results showed that there is room to derive tighter bounds. Due to the recursive calculations required to calculate the Bhattacharyya parameters at an arbitrary bit-channel, then the exact bounds are expected to be recursive in nature. For general BMS channels, it is more difficult to get closed form bounds as the polarization results in bit-channels with exponentially large alphabets.
Another issue is that we considered the construction of relaxed polar codes based on Ar{\i}kan's $2 \times 2$ polarization matrix. This construction can be readily extended to the general case of $l \times l$ polarization matrices, where $l \geq 3$ \cite{KoradaPolarConst}.
Also, it is noted that relaxing the good bit-channels results in rate-1 constituent codes. These bit-channels can be further concatenated with other codes to further reduce the error probability. In fact, the interleaved concatenation scheme for polar codes \cite{mahdavifar2014concatenated} adaptively concatenates better bit-channels with outer codes whose rates are higher than those concatenated with the worse bit-channels,  in order to maintain the target code rate or the target error performance of the concatenated code. When constructing concatenated relaxed polar codes, the adaptive concatenation scheme can be modified to take into account, or jointly optimize, the selection of the relaxed bit channels.

\section*{Acknowledgement}
The authors would like to sincerely thank the associate editor and the reviewers for their careful review of this paper and for their valuable comments which have improved its quality.

  \bibliographystyle{IEEEtran}

\bibliography{polar}

% Generated by IEEEtran.bst, version: 1.13 (2008/09/30)
\begin{thebibliography}{10}
\providecommand{\url}[1]{#1}
\csname url@samestyle\endcsname
\providecommand{\newblock}{\relax}
\providecommand{\bibinfo}[2]{#2}
\providecommand{\BIBentrySTDinterwordspacing}{\spaceskip=0pt\relax}
\providecommand{\BIBentryALTinterwordstretchfactor}{4}
\providecommand{\BIBentryALTinterwordspacing}{\spaceskip=\fontdimen2\font plus
\BIBentryALTinterwordstretchfactor\fontdimen3\font minus
  \fontdimen4\font\relax}
\providecommand{\BIBforeignlanguage}[2]{{%
\expandafter\ifx\csname l@#1\endcsname\relax
\typeout{** WARNING: IEEEtran.bst: No hyphenation pattern has been}%
\typeout{** loaded for the language `#1'. Using the pattern for}%
\typeout{** the default language instead.}%
\else
\language=\csname l@#1\endcsname
\fi
#2}}
\providecommand{\BIBdecl}{\relax}
\BIBdecl

\bibitem{ElkhamyRelaxed}
M.~El-Khamy, H.~Mahdavifar, G.~Feygin, J.~Lee, and I.~Kang, ``Relaxed channel
  polarization for reduced complexity polar coding,'' in \emph{2015 IEEE
  Wireless Communications and Networking Conference (WCNC)}, March 2015, pp.
  207--212.

\bibitem{arikan2009channel}
E.~Arikan, ``Channel polarization: A method for constructing capacity-achieving
  codes for symmetric binary-input memoryless channels,'' \emph{IEEE
  Transactions on Information Theory}, vol.~55, no.~7, pp. 3051--3073, 2009.

\bibitem{arikan2009rate}
E.~Arikan and E.~Telatar, ``On the rate of channel polarization,'' in
  \emph{2009 IEEE International Symposium on Information Theory (ISIT
  2009)}.\hskip 1em plus 0.5em minus 0.4em\relax IEEE, 2009, pp. 1493--1495.

\bibitem{Arikan2}
E.~Arikan, ``Source polarization,'' in \emph{2010 IEEE International Symposium
  on Information Theory Proceedings (ISIT)}.\hskip 1em plus 0.5em minus
  0.4em\relax IEEE, 2010, pp. 899--903.

\bibitem{Ab}
E.~Abbe, ``Randomness and dependencies extraction via polarization,'' in
  \emph{2011 Information Theory and Applications Workshop (ITA)}.\hskip 1em
  plus 0.5em minus 0.4em\relax IEEE, 2011, pp. 1--7.

\bibitem{polarBICM}
H.~Mahdavifar, M.~El-Khamy, J.~Lee, and I.~Kang, ``Polar coding for
  bit-interleaved coded modulation,'' \emph{IEEE Transactions on Vehicular
  Technology}, 2015.

\bibitem{MV}
H.~Mahdavifar and A.~Vardy, ``Achieving the secrecy capacity of wiretap
  channels using polar codes,'' \emph{IEEE Transactions on Information Theory},
  vol.~57, no.~10, pp. 6428--6443, 2011.

\bibitem{STY}
E.~\c{S}a\c{s}o\u{g}lu, E.~Telatar, and E.~Yeh, ``Polar codes for the two-user
  binary-input multiple-access channel,'' in \emph{2010 IEEE Information Theory
  Workshop (ITW)}.\hskip 1em plus 0.5em minus 0.4em\relax IEEE, 2010, pp. 1--5.

\bibitem{AT2}
E.~Abbe and E.~Telatar, ``Polar codes for the m-user multiple access channel,''
  \emph{IEEE Transactions on Information Theory}, vol.~58, no.~8, pp.
  5437--5448, 2012.

\bibitem{polar_MAC}
H.~Mahdavifar, M.~El-Khamy, J.~Lee, and I.~Kang, ``Techniques for polar coding
  over multiple access channels,'' in \emph{2014 48th Annual Conference on
  Information Sciences and Systems (CISS)}, March 2014, pp. 1--6.

\bibitem{polarBroadcast}
N.~Goela, E.~Abbe, and M.~Gastpar, ``Polar codes for broadcast channels,''
  \emph{IEEE Transactions on Information Theory}, vol.~61, no.~2, pp. 758--782,
  Feb 2015.

\bibitem{KoradaPolarConst}
S.~Korada, E.~\c{S}a\c{s}o\u{g}lu, and R.~Urbanke, ``Polar codes:
  Characterization of exponent, bounds, and constructions,'' \emph{IEEE
  Transactions on Information Theory}, vol.~56, no.~12, pp. 6253--6264, Dec
  2010.

\bibitem{compound}
H.~Mahdavifar, M.~El-Khamy, J.~Lee, and I.~Kang, ``Compound polar codes,'' in
  \emph{2013 Information Theory and Applications Workshop (ITA)}.\hskip 1em
  plus 0.5em minus 0.4em\relax IEEE, 2013, pp. 1--6.

\bibitem{mahdavifar2014concatenated}
------, ``Performance limits and practical decoding of interleaved
  {Reed-Solomon} polar concatenated codes,'' \emph{IEEE Transactions on
  Communications}, vol.~62, no.~5, pp. 1406--1417, 2014.

\bibitem{hassani2014universal}
S.~H. Hassani and R.~Urbanke, ``Universal polar codes,'' in \emph{2014 IEEE
  International Symposium on Information Theory (ISIT)}.\hskip 1em plus 0.5em
  minus 0.4em\relax Ieee, 2014, pp. 1451--1455.

\bibitem{Gross2}
C.~Leroux, A.~J. Raymond, G.~Sarkis, I.~Tal, A.~Vardy, and W.~J. Gross,
  ``Hardware implementation of successive-cancellation decoders for polar
  codes,'' \emph{Journal of Signal Processing Systems}, vol.~69, no.~3, pp.
  305--315, 2012.

\bibitem{Gross1}
C.~Leroux, A.~J. Raymond, G.~Sarkis, and W.~J. Gross, ``A semi-parallel
  successive-cancellation decoder for polar codes,'' \emph{Signal Processing,
  IEEE Transactions on}, vol.~61, no.~2, pp. 289--299, 2013.

\bibitem{MultiDimFast}
H.~Mahdavifar, M.~El-Khamy, J.~Lee, and I.~Kang, ``Fast multi-dimensional polar
  encoding and decoding,'' in \emph{2014 Information Theory and Applications
  Workshop (ITA)}, Feb 2014, pp. 1--5.

\bibitem{alamdar2011simplified}
A.~Alamdar-Yazdi and F.~R. Kschischang, ``A simplified successive-cancellation
  decoder for polar codes,'' \emph{IEEE communications letters}, vol.~15,
  no.~12, pp. 1378--1380, 2011.

\bibitem{BPKailath}
T.~Kailath, ``The divergence and {Bhattacharyya} distance measures in signal
  selection,'' \emph{IEEE Transactions on Communication Technology}, vol.~15,
  no.~1, pp. 52--60, February 1967.

\bibitem{hassani2010scaling}
S.~H. Hassani and R.~Urbanke, ``On the scaling of polar codes: I. the behavior
  of polarized channels,'' in \emph{2010 IEEE International Symposium on
  Information Theory Proceedings (ISIT)}.\hskip 1em plus 0.5em minus
  0.4em\relax IEEE, 2010, pp. 874--878.

\bibitem{TalVardy13}
I.~Tal and A.~Vardy, ``How to construct polar codes,'' \emph{IEEE Transactions
  on Information Theory}, vol.~59, no.~10, pp. 6562--6582, Oct 2013.

\bibitem{Trifonov}
P.~Trifonov, ``Efficient design and decoding of polar codes,'' \emph{IEEE
  Transactions on Communications}, vol.~60, no.~11, pp. 3221--3227, November
  2012.

\bibitem{richardson2008modern}
T.~Richardson and R.~L. Urbanke, \emph{Modern coding theory}.\hskip 1em plus
  0.5em minus 0.4em\relax Cambridge University Press, 2008.

\bibitem{TVlist}
I.~Tal and A.~Vardy, ``List decoding of polar codes,'' \emph{IEEE Transactions
  on Information Theory}, vol.~61, no.~5, pp. 2213--2226, 2015.

\bibitem{NiuCRC}
K.~Niu and K.~Chen, ``{CRC}-aided decoding of polar codes,'' \emph{IEEE
  Communications Letters}, vol.~16, no.~10, pp. 1668--1671, October 2012.

\bibitem{NiuSurvey}
K.~Niu, K.~Chen, J.~Lin, and Q.~Zhang, ``Polar codes: Primary concepts and
  practical decoding algorithms,'' \emph{IEEE Communications Magazine},
  vol.~52, no.~7, pp. 192--203, July 2014.

\end{thebibliography}

\section*{Appendix}

For any DMC $W$, let $E(W)$ denote the probability of error of $W$ under ML decoder. Let $W: \sX \rightarrow \sY$ be a BMS channel, where $\sX$ is the binary alphabet. Then
$$
E(W) = \sum_{y \in \sY} \frac{1}{2} \min\left\{W(y|0),W(y|1)\right\}
$$
Let the channel $W$ be polarized using the Ar{\i}kan's Butterfly. Let the polarized bit-channels be denoted by $W^+$ and $W^-$. Then it is known that
\be{operation1}
W^-(y_1,y_2 | u_1) = \frac{1}{2} \sum_{u_2 \in \sX} W(y_1 | u_1 \oplus u_2) W(y_2|u_2)
\ee
and
\be{operation2}
W^+(y_1,y_2,u_1 | u_2) = \frac{1}{2} W(y_1 | u_1 \oplus u_2) W(y_2|u_2)
\ee
\begin{lemma}
For any BMS $W$,
\be{w-}
E(W^-) = 2 E(W) - 2E(W)^2
\ee
and
\be{w+}
E(W^+) \geq 2E(W)^2
\ee
\end{lemma}
\begin{proof}
Suppose that the size of output alphabet $\sY$ is $M$. Let $\sY=\left\{y_1,y_2,...,y_M\right\}$. Then for $i=1,2,\dots,M$, let
$$
a_i = \min\left\{W(y|0),W(y|1)\right\}
$$
and
$$
b_i = \max\left\{W(y|0),W(y|1)\right\}
$$
Then
\be{eq1}
E(W) = \frac{1}{2} \sum_{i=1}^M a_i
\ee
and
\be{eq2}
\sum_{i=1}^M a_i + \sum_{i=1}^M b_i = 2
\ee
The size of output alphabet of $W^-$ is $M^2$. For any pair $(y_i,y_j) \in \sY^2$,
\begin{align*}
&\left\{W^-(y_i,y_j|0),W^-(y_i,y_j|1)\right\} \\
& = \left\{\frac{1}{2}(a_i a_j + b_i b_j),\frac{1}{2}(a_i b_j + a_j b_i)\right\}
\end{align*}
Notice that $a_i a_j + b_i b_j \geq a_i b_j + a_j b_i$. Therefore,
\begin{align*}
E(W^-) &= \frac{1}{4}\sum_{i=1}^M \sum_{j=1}^M (a_i b_j + a_j b_i) = \frac{1}{2}(\sum_{j=1}^M a_i)(\sum_{j=1}^M b_j) \\ & = E(W)(2-2E(W))
\end{align*}
where the last equality follows by \eqref{eq1} and \eqref{eq2}. This proves the first part of the lemma.

The size of the output alphabet of $W^+$ is $2M^2$. For any pair $(y_i,y_j) \in \sY^2$, there are two corresponding elements in the output alphabet of $W^+$. Then
\eq{\begin{split} &
\Bigl\{\left\{W^+(y_i,y_j,0|0),W^+(y_i,y_j,0|1)\right\} \\ &
,\left\{W^+(y_i,y_j,1|0),W^+(y_i,y_j,1|1)\right\}\Bigr\} =  \\
& \quad
\Bigl\{\left\{\frac{1}{2}a_ia_j,\frac{1}{2}b_ib_j\right\},\left\{\frac{1}{2}a_ib_j,\frac{1}{2}a_jb_i\right\}\Bigr\}
\end{split}}
Then
\begin{align*}
2E(W^+) &= \frac{1}{2}\sum_{i=1}^M \sum_{j=1}^M a_i a_j + \frac{1}{2} \sum_{i=1}^M \sum_{j=1}^M \min\left\{a_ib_j,a_jb_i\right\} \\
& \geq \sum_{i=1}^M \sum_{j=1}^M a_i a_j = (\sum_{i=1}^M a_i)^2 = 4 E(W)^2
\end{align*}
which proves the second part of the lemma.
\end{proof}

\end{document}